\documentclass[11pt]{article}
\usepackage{amsmath,amssymb,amsthm,amsfonts,latexsym,bbm,xspace,graphicx,float,mathtools}

\usepackage[backref,colorlinks,citecolor=blue,bookmarks=true]{hyperref}

\usepackage{amssymb}
\usepackage{amsfonts}
\usepackage{amsmath, amsthm}

\usepackage{latexsym}
\usepackage{epsfig}
\usepackage{bm}
\usepackage{times}
\usepackage{color}

\newcommand{\ignore}[1]{}
\newenvironment{prevproof}[2]{\noindent {\em {Proof of {#1}~\ref{#2}:}}}{\hfill{$\square$}\vskip \belowdisplayskip}

\newtheorem{claim}{Claim}

\newtheorem{remark}[claim]{Remark}
\newtheorem{lemma}[claim]{Lemma}
\newtheorem{theorem}{Theorem}
\newtheorem{definition}{Definition}

\newcommand{\Poi}{\mathrm {Poi}}

\newcommand{\dtv}{d_{\mathrm TV}}
\newcommand{\dk}{d_{\mathrm K}}

\newcommand{\R}{{\bf R}}

\newcommand{\E}{{\bf E}}

\newcommand{\Var}{\mathrm{Var}}

\newcommand{\eps}{\epsilon}
\newcommand{\littlesum}{\mathop{{\textstyle \sum}}}
\newcommand{\littleprod}{\mathop{{\textstyle \prod}}}

\newcommand{\poly}{\mathrm{poly}}

\newcommand{\bit}[1]{{\langle{#1}\rangle}}

   \newcommand{\remove}[1]{\par $<<${\it removed part}$>>$}
   \newcommand{\old}[1]{}

\newcommand{\costasnote}[1]{#1}

\newcommand{\blue}[1]{{{#1}}}

\newcommand{\rocconote}[1]{{{#1}}}

\newcommand{\green}[1]{{\color{green} {#1}}}
\newcommand{\rnote}[1]{\footnote{{\green{\bf [[Rocco: {#1}\bf ]] }}}}

\begin{document}

%

\title{Learning Poisson Binomial Distributions
}

\author{Constantinos Daskalakis\thanks{Research supported by a Sloan Foundation Fellowship, a Microsoft Research Faculty Fellowship, and NSF Award CCF-
0953960 (CAREER) and CCF-1101491.}
\\
MIT\\
{\tt costis@csail.mit.edu}\\
\and
Ilias Diakonikolas\thanks{Research supported by a Simons Foundation Postdoctoral Fellowship. Some of this work
was done while at Columbia University, supported by NSF grant CCF-0728736, and by an Alexander S. Onassis Foundation
Fellowship.}\\
University of Edinburgh\\
{\tt ilias.d@ed.ac.uk}\\
\and
Rocco A. Servedio
\thanks{Supported by NSF grants CNS-0716245, CCF-1115703, and
CCF-1319788 and by DARPA award HR0011-08-1-0069.}\\
Columbia University\\
{\tt rocco@cs.columbia.edu}}

\setcounter{page}{0}

\maketitle

\thispagestyle{empty}

\begin{abstract}
We consider a basic problem in unsupervised learning:
learning an unknown \emph{Poisson Binomial Distribution}.
A Poisson Binomial Distribution (PBD) over $\{0,1,\dots,n\}$
is the distribution
of a sum
of $n$ independent Bernoulli
random variables which may have arbitrary, potentially non-equal,
expectations.  These distributions were first studied by S. Poisson in 1837 \cite{Poisson:37} and are a natural
$n$-parameter generalization of the familiar Binomial Distribution.
Surprisingly, prior to our work this basic learning problem
was poorly understood, and known results for it were far from
optimal.

We essentially settle the complexity of the learning problem for this
basic class of distributions.
As our first main result we give a highly efficient algorithm which learns to $\eps$-accuracy (with respect to the total variation distance) using $\tilde{O}(1/\eps^3)$ samples \emph{independent of $n$}.  The running time of the algorithm is \emph{quasilinear} in the
size of its input data, i.e., $\tilde{O}(\log(n)/\eps^3)$ bit-operations.\footnote{We
write $\tilde{O}(\cdot)$ to hide factors which are polylogarithmic
in the argument to $\tilde{O}(\cdot)$;
thus, for example, $\tilde{O}(a \log b)$ denotes a quantity
which is {$O(a \log b \cdot \log^c(a \log b))$} for some
absolute constant $c$.}
(Observe that each draw from the distribution is a $\log(n)$-bit string.)
{Our second main result is a {\em proper} learning algorithm that learns to $\eps$-accuracy using $\tilde{O}(1/\eps^2)$ samples, and runs in time
$(1/\eps)^{\poly (\log (1/\eps))} \cdot \log n$.}
 This sample complexity is nearly optimal,
since any algorithm {for this problem} must use $\Omega(1/\eps^2)$ samples.
We also give positive and negative results for some extensions of this  learning problem to weighted sums of independent Bernoulli random variables.

\end{abstract}

\section{Introduction}

We begin by considering a somewhat fanciful scenario:  You are the manager of an independent weekly newspaper in a
city of $n$ people.  Each week the $i$-th inhabitant of the city independently picks up a copy of your paper with probability $p_i$.  Of course you do not know the values $p_1,\dots,p_n$; each week you only see the total number of papers that have been picked up.  For many reasons (advertising, production, revenue analysis, etc.) you would like to have a detailed ``snapshot'' of the probability distribution (pdf) describing how many readers you have each week.
{\emph{Is there an efficient algorithm to construct a high-accuracy approximation of the pdf from a number of observations that is {\em independent} of the population $n$?} We show that the answer is ``yes.''}

A \emph{Poisson Binomial Distribution} of order $n$ is
the distribution of a sum \[X=\sum_{i=1}^n X_i,\]
where $X_1,\dots,X_n$ are independent Bernoulli (0/1)
random variables. The expectations $(\E[X_i]=p_i)_i$
need not  all be the same, and thus these distributions
generalize the Binomial distribution ${\rm Bin}(n,p)$ and,
indeed, comprise a much
richer class of distributions. (See Section~\ref{sec:relatedwork} below.)
It is believed that Poisson \cite{Poisson:37} was the first to consider
this extension of the Binomial distribution\footnote{We thank Yuval Peres and Sam Watson for this information~\cite{PW11personal}.} and the distribution
is sometimes referred to as ``Poisson's Binomial Distribution'' in his
honor; we shall simply call these distributions PBDs.

PBDs are one of the most basic classes of discrete distributions;
indeed, they are arguably the simplest $n$-parameter
probability distribution that has some nontrivial structure.
As such they have been intensely studied in probability and statistics
(see Section~\ref{sec:relatedwork}) and arise in many settings; for
example, we note here that tail bounds on PBDs form an important special
case of Chernoff/Hoeffding bounds \cite{Chernoff:52,Hoeffding:63,DP09}.
In application domains, PBDs have many uses in research areas
such as survey sampling, case-control studies, and survival analysis, see e.g., \cite{ChenLiu:97} for a survey of the many uses of these distributions in applications.
Given the simplicity and ubiquity of these
distributions, it is quite surprising 
that the problem of \emph{density estimation} for PBDs (i.e., learning
an unknown PBD from independent samples) is not well understood
in the statistics or learning theory literature.
{\em This is the problem we consider, and essentially settle, in this paper.}

We work in a natural PAC-style model of learning an unknown discrete probability
distribution which is essentially the model of \cite{KMR+:94}.
In this learning framework for our problem, the learner is provided with the value of $n$ and with independent samples drawn from an unknown PBD $X$.  Using these samples, the learner must with probability at least $1 - \delta$
output a hypothesis distribution $\hat{X}$ such that
the total variation distance $\dtv(X,\hat{X})$ is at most $\eps$, where $\eps,\delta > 0$ are accuracy and confidence parameters that are provided to the learner.\footnote{\cite{KMR+:94} used the Kullback-Leibler divergence as their distance measure but we find it more natural to use variation distance.}
A \emph{proper} learning algorithm in this framework outputs a distribution that is itself
a Poisson Binomial Distribution, i.e., a vector $\hat{p}=(\hat{p}_1,\dots,\hat{p}_n)$  which describes the hypothesis PBD  $\hat{X}=\sum_{i=1}^n \hat{X}_i$ where $\E[\hat{X}_i]=\hat{p}_i$.

\subsection{Our results.}

{Our main result is an efficient algorithm for learning PBDs from
$\tilde{O}(1/\eps^2)$ many samples independent of $[n].$
Since PBDs are an $n$-parameter family of distributions over the domain $[n]$,
we view such a tight bound as a surprising result.} We prove:

\begin{theorem}   [\bf Main Theorem] \label{thm:main}
Let $X = \sum_{i=1}^n X_i$ be an unknown PBD.
\begin{enumerate}
  \item {\bf [Learning PBDs from constantly many samples]} There is an algorithm with the following properties:
given $n, \eps,\delta$
and access to independent draws from $X$, the algorithm uses
$$\tilde{O}\left( (1/\eps^3)  \cdot \log(1/\delta) \right)$$ samples from $X$,
performs $$\tilde{O} \left( (1/\eps^3) \cdot \log n  \cdot \log^2 {1 \over \delta} \right)$$ bit operations,
and
with probability at least $1-\delta$ outputs a (succinct description of a) distribution $\hat{X}$ over $[n]$
which is such that $\dtv(\hat{X},X) \leq \eps.$

  \item {\bf [Properly learning PBDs from constantly many samples]} There is an algorithm with the following properties:
given $n,\eps,\delta$
and access to independent draws from $X$, the algorithm uses
$$\tilde{O}(1/\eps^{{2}})\cdot \log(1/\delta) $$
samples from $X$, performs
$$(1/\eps)^{O\left( \log^2(1/\eps) \right)} \cdot \tilde{O} \left( \log n  \cdot \log {1 \over \delta} \right)$$
bit operations, and with probability at least $1-\delta$ outputs
a (succinct description of a) vector $\hat{p}=(\hat{p}_1,\dots,\hat{p}_n)$ defining a PBD $\hat{X}$
such that $\dtv(\hat{X},X) \leq \eps.$
\end{enumerate}
\end{theorem}

We note that, since every sample drawn from $X$ is a $\log(n)$-bit string,
{for constant $\delta$}
the number of bit-operations performed by our first algorithm is
\emph{quasilinear} in the length of its input.  Moreover, the sample complexity of
both algorithms is close to optimal, since
$\Omega(1/\eps^2)$ samples are required even to distinguish the (simpler) Binomial
distributions ${\rm Bin}(n,1/2)$ and ${\rm Bin}(n,1/2 + \eps/\sqrt{n})$,
which have total variation distance $\Omega(\eps).$ {Indeed, in view of this observation, our second algorithm is essentially sample-optimal.}

Motivated by these strong learning results for PBDs, we also consider learning a more general class of distributions, namely distributions of the form
$X = \sum_{i=1}^n w_i X_i$ which are \emph{weighted} sums of independent
Bernoulli random variables.  We give an algorithm which uses $O(\log n)$
samples and runs in $\poly(n)$ time if there are only constantly many
different weights in the sum:

\begin{theorem}
[\bf Learning sums of weighted
independent Bernoulli random variables]
\label{thm:linearupper}Let $X = \sum_{i=1}^n a_i X_i$ be a weighted sum of unknown independent Bernoullis such that there are at most $k$ different values among
$a_1,\dots,a_n.$  Then there is an algorithm with the following properties:  given $n,\eps,\delta,$  $a_1,\dots,a_n$ and access to independent draws from $X$, it uses
$$\widetilde{O}(k/\eps^2) \cdot \log(n)  \cdot \log(1/\delta)$$
samples from $X$,
runs in time
$$\poly \left( n^k \cdot \eps^{-k\log^2(1/\eps)} \right) \cdot \log(1/\delta),$$
and with probability at least $1-\delta$ outputs a hypothesis vector $\hat{p} \in [0,1]^n$ defining independent Bernoulli random variables $\hat{X}_i$ with {$\E[\hat{X}_i]=\hat{p}_i$} such that $\dtv(\hat{X},X) \leq \eps,$ where $\hat{X}=\sum_{i=1}^n a_i \hat{X}_i$.
\end{theorem}

\ignore{Note that setting all $a_i$'s to 1
Theorem~\ref{thm:linearupper} gives a weaker result than
Theorem~\ref{thm:main} in terms of running time and sample complexity.}
To complement Theorem~\ref{thm:linearupper}, we also show that if there are many distinct weights in
the sum, then even for weights with a very simple structure any learning algorithm must use many samples:

\begin{theorem} [\bf Sample complexity lower bound for learning sums of
weighted independent Bernoullis] \label{thm:linearlower}
Let $X=\sum_{i=1}^n i \cdot X_i$ be a weighted sum of unknown independent Bernoullis
(where the $i$-th weight is simply $i$). Let $L$ be any learning algorithm which, given $n$
and access to independent draws from $X$, outputs a hypothesis distribution $\hat{X}$ such that
$\dtv(\hat{X},X) \leq 1/25$ with probability at least
$e^{-o(n)}.$ Then $L$ must use $\Omega(n)$ samples.
\end{theorem}



\subsection{Related work.} \label{sec:relatedwork}
At a high level, there has been a recent surge of interest in
the theoretical computer science  community on
fundamental algorithmic problems involving basic types of
probability distributions,
see e.g., \cite{KMV:10, MoitraValiant:10, BelkinSinha:10,
ValiantValiant:11} and other recent papers; our work may be considered as an extension of this theme.
More specifically, there is a broad literature in probability theory studying various properties of PBDs; see \cite{Wang93} for an accessible introduction to some of this work.  In particular,
{many results  study approximations to the Poisson Binomial distribution via simpler distributions. In a well-known result, Le Cam \cite{LeCam:60} shows that for any PBD $X=\sum_{i=1}^n X_i$
with $\E[X_i]=p_i$, it holds that
$$
\dtv \left( X,\Poi \big( \littlesum_{i=1}^n p_i \big) \right) \leq 2 \littlesum_{i=1}^n p_i^2,
$$
where $\Poi(\lambda)$ is  the Poisson distribution with parameter $\lambda$.  Subsequently many other proofs of this result and similar ones were given using a range of different techniques;
\cite{HodgesLeCam:60,Chen:74,DP:86,BHJ:92} is a sampling of work along these lines, and Steele \cite{Steele:94} gives an extensive list of relevant references.    Much work has also been done on approximating PBDs by normal distributions (see e.g., \cite{berry,esseen,Mikhailov:93,Volkova:95}) and by Binomial distributions (see e.g., \cite{Ehm91,Soon:96,Roos:00}). These results provide structural information about PBDs that can be well-approximated via simpler distributions, but fall short of our goal of obtaining approximations of an unknown PBD up to {\em arbitrary accuracy}. Indeed, the approximations obtained in the probability literature (such as the Poisson, Normal and Binomial approximations)  typically depend only on the first few moments of the target PBD.  This is attractive from a learning perspective because it is possible to efficiently estimate  such moments from random samples, but higher moments are crucial for arbitrary approximation~\cite{Roos:00}.

Taking a different perspective, it is easy to show (see Section~2 of \cite{KeilsonGerber:71}) that every PBD is a unimodal distribution over $[n]$.  (Recall that a distribution $p$ over $[n]$ is unimodal if there is a value $\ell \in \{0,\dots,n\}$
such that $p(i) \leq p(i+1)$ for $i \leq \ell$ and $p(i) \geq p(i+1)$ for $i > \ell.$)  The learnability of general unimodal distributions over $[n]$ is well understood:  Birg\'e
\cite{Birge:87,Birge:97} has given a computationally efficient
algorithm that can learn any unimodal distribution over $[n]$ to variation
distance $\eps$ from $O(\log(n)/\eps^3)$ samples, and has shown that any
algorithm must use $\Omega(\log(n)/\eps^3)$ samples. (The \cite{Birge:87,Birge:97} upper and lower
bounds are stated for continuous unimodal distributions,
but the arguments are easily adapted to the discrete case.) Our main result, Theorem~\ref{thm:main}, shows that the additional PBD
assumption can be leveraged to obtain sample complexity \emph{independent
of $n$} with a computationally highly efficient algorithm.

So, how might one leverage the structure of PBDs to remove $n$ from the sample complexity? A first observation is that a PBD assigns $1-\epsilon$ of its mass to $O_\eps(\sqrt{n})$ points. So
one could draw samples to (approximately) identify these points and then try to estimate the probability assigned to each such point,
but clearly such an approach, if followed na\"ively, would give $\poly(n)$ sample complexity. Alternatively, one  could run Birg\'e's algorithm on the restricted support of size $O_\eps( \sqrt{n})$, but that will not improve the asymptotic sample complexity. A different approach would be to construct a small $\epsilon$-cover (under the total variation distance) of the space of all PBDs on $n$ variables. Indeed, if such a cover has size $N$, it can be shown (see Lemma~\ref{lem:log-cover-size}
in Section~\ref{sec:positive}, or Chapter 7 of \cite{DL:01})) that a target PBD can be
learned from
$O(\log (N)/\epsilon^2)$ samples. Still it is easy to argue that any cover needs to have size $\Omega(n)$, so this approach too gives a $\log(n)$ dependence in the sample complexity.

Our approach, which removes $n$ completely from the sample complexity, requires a refined understanding of the structure of the set of all PBDs on $n$ variables, in fact one that is more refined than the understanding provided by the aforementioned results (approximating a PBD by a Poisson, Normal, or Binomial distribution). We give an outline of the approach in the next section.}

\subsection{Our approach.}

The starting point of our algorithm for learning PBDs is a theorem of~\cite{DP:oblivious11,Daskalakis:anonymous08full} that gives detailed information about the structure of a small $\epsilon$-cover (under the total variation distance) of the space of all PBDs on $n$ variables (see Theorem~\ref{thm: sparse cover theorem}). Roughly speaking, this result says that every PBD is either close to a PBD whose support is sparse, or is close to a translated ``heavy'' Binomial distribution.  Our learning algorithm exploits this structure of the cover\ignore{ to close in on the information that is absolutely necessary to approximate an unknown PBD. In particular,}; it has two subroutines corresponding to these two different types of distributions that the cover {contains}. First, assuming that the target PBD is close to a sparsely supported distribution, it runs Birg\'e's unimodal distribution learner over a carefully selected subinterval of $[n]$ to construct a hypothesis $H_S$; the (purported) sparsity of the distribution makes it possible for this algorithm to use $\tilde{O}(1/\eps^3)$ samples independent of $n$. Then, assuming that the target PBD is close to a translated ``heavy'' Binomial distribution, the algorithm constructs a hypothesis Translated Poisson Distribution $H_P$ \cite{Rollin:07} whose mean and variance match the estimated mean and variance of the target PBD; we show that $H_P$ is close to the target PBD if the target PBD is not close to any sparse distribution in the cover.  At this point the algorithm has two hypothesis distributions, $H_S$ and $H_P$, one of which should be good; it remains to select one as the final output hypothesis.  This is achieved using a form of ``hypothesis testing'' for probability distributions.

The above sketch captures the main ingredients of Part (1) of Theorem~\ref{thm:main}, but additional work needs to be done to get the proper learning algorithm of Part (2).
{For the non-sparse case,}
first note {that} the Translated Poisson hypothesis $H_P$
is not a PBD.  Via a sequence of transformations we are able to show that the
Translated Poisson hypothesis $H_P$ can be converted to a Binomial
distribution $\mathrm{Bin}(n',p)$
for some $n' \leq n.$  {To handle the sparse case, we use an alternate
learning approach: instead of using Birg\'e's unimodal algorithm (which
would incur a sample complexity of $\Omega(1/\eps^3)$), we {first}
show that, in this case, there exists an
{efficiently constructible}
$O(\eps)$-cover of size
$(1/\eps)^{O(\log^2(1/\eps))},$ {and then apply a general learning
result that we now describe.}}

The general learning result that we use
{(Lemma~\ref{lem:log-cover-size})}
is the following:
We show that for any class ${\cal S}$ of target distributions, if ${\cal S}$
has an $\eps$-cover of size $N$ then
there is a generic algorithm for learning an unknown distribution
from ${\cal S}$ to accuracy ${O(\eps)}$
that uses $O((\log N)/\eps^2)$ samples.
Our approach is rather similar to the algorithm of
\cite{DL:01} for choosing a density estimate (but different in some details);
it works by carrying out a tournament that
matches every pair of distributions in the cover against each
other. Our analysis shows that with high probability some
$\eps$-accurate distribution in the cover will survive the tournament
undefeated, and that any undefeated distribution will with high probability
be $O(\eps)$-accurate.

{Applying this general result to the $O(\eps)$-cover of size
$(1/\eps)^{O(\log^2(1/\eps))}$  described above,
we obtain a PBD that is $O(\eps)$-close to the target (this accounts for the
increased running time in Part (2) versus Part (1)).} We stress that for
both the non-proper and proper learning algorithms
sketched above, many technical subtleties and challenges arise in implementing the high-level plan given above, requiring a careful and detailed analysis.

We prove Theorem~\ref{thm:linearupper} using
the general approach of Lemma~\ref{lem:log-cover-size}
specialized to weighted sums of independent Bernoullis with constantly many
distinct weights.
We show how the tournament can be implemented efficiently for the
class ${\cal S}$ of weighted sums of independent Bernoullis with constantly
many distinct weights, {and thus obtain Theorem~\ref{thm:linearupper}.}
Finally, the lower bound of Theorem~\ref{thm:linearlower} is proved by a
direct information-theoretic argument.

\subsection{Preliminaries.} \label{sec:prelim}

\paragraph{Distributions.} For a distribution $X$ supported on $[n]=\{0,1,\dots,n\}$ we write $X(i)$ to denote
the value $\Pr[X=i]$ of the probability density function (pdf) at point $i$, and $X(\leq i)$ to denote the value $\Pr[X \leq i]$ of the cumulative density function (cdf) at point $i$.
For $S \subseteq [n]$, we write $X(S)$ to denote $\sum_{i \in S}X(i)$ and $X_S$ to denote
the conditional distribution of $X$ restricted to $S.$ Sometimes we write $X(I)$ and $X_I$ for a subset $I \subseteq [0,n]$, meaning $X(I \cap [n])$ and $X_{I \cap [n]}$ respectively.

\paragraph{Total Variation Distance.}Recall that the {\em total variation distance} between two distributions
$X$ and $Y$ over a finite domain $D$ is
\begin{eqnarray*}
\dtv\left(X,Y \right) &:=& (1/2)\cdot \littlesum_{\alpha
\in D}{|X(\alpha)-Y(\alpha)|} = \max_{S \subseteq D}[X(S)-Y(S)].
\end{eqnarray*}
Similarly, if $X$ and
$Y$ are two random variables ranging over a finite set, their total
variation distance $\dtv(X,Y)$ is defined as the total variation
distance between their distributions.
For convenience, we will often blur the distinction between a random variable and its distribution.\ignore{Another useful notion of distance between
distributions/random variables is the {\em Kolmogorov
distance}.\rnote{Do we use this now in this paper?} For two distributions $X$ and $Y$ supported
on $\R$, their Kolmogorov distance is
$\dk\left(X,Y \right):= \sup_{x \in \R} \left|
X(\leq x)-Y(\leq x) \right|.$  }\ignore{Similarly,
if $X$ and $Y$ are two random variables ranging over a subset of $\R$,
their Kolmogorov distance, denoted $\dk(X,Y),$ is the Kolmogorov distance
between their distributions.}\ignore{ If two distributions $X$ and $Y$
are supported on a finite subset of $\R$ we obtain
immediately that $\dk\left(X,Y
\right) \le \dtv\left(X,Y \right).$\rnote{Should there be a 2 here -- I think not?  Do we need this
observation at all?}}

\paragraph{Covers.}Fix a finite domain $D$, and let ${\cal P}$ denote some set of
distributions over $D.$ Given $\delta > 0$, a subset ${\cal Q} \subseteq
{\cal P}$ is said to be a \emph{$\delta$-cover of ${\cal P}$} (w.r.t. the
total variation distance) if for every distribution $P$ in ${\cal P}$ there
exists some distribution $Q$ in ${\cal Q}$ such that
$\dtv(P,Q) \leq \delta.$
{We sometimes say that distributions $P,Q$ are
\emph{$\delta$-neighbors} if $\dtv(P,Q) \leq \delta$.} {If this holds, we also say that $P$ is $\delta$-close to $Q$ and vice versa.}

\paragraph{Poisson Binomial Distribution.}
A {\em Poisson binomial distribution of order $n \in \mathbb{N}$} is a sum $\sum_{i=1}^n X_i$ of $n$ mutually independent Bernoulli random variables $X_1,\ldots,X_n$. We denote the set of all Poisson binomial distributions of order $n$ by ${\cal S}_{n}$ and, if $n$ is clear from context, just ${\cal S}$.

A Poisson binomial distribution $D  \in {\cal S}_n$ can be represented uniquely as a vector $(p_i)_{i=1}^n$ satisfying $0\le p_1 \le p_2 \le \ldots \le p_n \le 1$. To go from $D \in {\cal S}_n$ to its corresponding vector, we find a collection $X_1,\ldots,X_n$ of mutually independent Bernoullis such that $\sum_{i=1}^n X_i$ is distributed according to $D$ and $\E[X_1]\le \ldots \le \E[X_n]$. (Such a collection exists by the definition of a Poisson binomial distribution.) Then we set $p_i = \E[X_i]$ for all $i$.  Lemma 1 of~\cite{DP:cover} shows that the resulting vector $(p_1,\ldots,p_n)$ is unique.

We denote by ${\rm PBD}(p_1,\ldots,p_n)$ the distribution of the sum $\sum_{i=1}^n X_i$ of mutually independent indicators  $X_1,\ldots,X_n$ with expectations $p_i=\E[X_i]$, for all $i$. Given the above discussion ${\rm PBD}(p_1,\ldots,p_n)$ is unique up to permutation of the $p_i$'s. We also sometimes write $\{X_i\}$ to denote the distribution of $\sum_{i=1}^n X_i.$ Note the difference between $\{X_i\}$, which refers to the distribution of $\sum_i X_i$, and $\{X_i\}_i$, which refers to the underlying collection of mutually independent Bernoulli random variables.

\paragraph{Translated Poisson Distribution.} We will make use of the translated Poisson distribution for approximating the Poisson Binomial distribution. We define the translated Poisson distribution, and state a known result on how well it approximates the Poisson Binomial distribution.
\begin{definition}[\cite{Rollin:07}]
We say that an integer random variable $Y$ is distributed according to the {\em translated Poisson distribution with parameters $\mu$ and $\sigma^2$}, denoted $TP(\mu, \sigma^2)$,
iff $Y$ can be written as
\[Y = \lfloor \mu-\sigma^2\rfloor + Z,\]
where $Z$ is a random variable distributed according to ${\rm Poisson}(\sigma^2+ \{\mu-\sigma^2\})$, where
$\{\mu-\sigma^2\}$ represents the fractional part of $\mu-\sigma^2$.
\end{definition}
The following lemma gives a useful bound on the variation distance between a Poisson Binomial
Distribution and a suitable translated Poisson distribution.  Note that if the variance of the Poisson
Binomial Distribution is large, then the lemma gives a strong bound.

\begin{lemma}[see (3.4) of \cite{Rollin:07}] \label{lem:translated Poisson approximation}
Let $J_1,\ldots,J_n$ be independent random indicators with $\E[J_i]=p_i$. Then
$$\dtv\left(\sum_{i=1}^nJ_i, TP(\mu, \sigma^2)\right) \le \frac{\sqrt{\sum_{i=1}^np_i^3(1-p_i)}+2}{\sum_{i=1}^np_i(1-p_i)},$$
where $\mu=\sum_{i=1}^np_i$ and $\sigma^2 = \sum_{i=1}^np_i(1-p_i)$.
\end{lemma}

The following bound on the total variation distance between translated Poisson distributions will be useful.

\begin{lemma}[Lemma~2.1 of \cite{BarbourLindvall}] \label{lem: variation distance between translated Poisson distributions}
For $\mu_1, \mu_2 \in \mathbb{R}$ and $\sigma_1^2, \sigma_2^2 \in \mathbb{R}_+$ with $\lfloor \mu_1-\sigma_1^2 \rfloor \le \lfloor \mu_2-\sigma_2^2 \rfloor$, we have
$$\dtv(TP(\mu_1,\sigma_1^2), TP(\mu_2,\sigma_2^2)) \le \frac{|\mu_1-\mu_2|}{\sigma_1}+\frac{|\sigma_1^2-\sigma_2^2|+1}{\sigma_1^2}.$$
\end{lemma}

\ignore{
%
%
}

\ignore{We
write $f(\overline{p})(\alpha)$ to denote the probability of outcome
$\alpha$ under distribution $f(\overline{p}).$ Finally, for $\ell \in
\mathbb{Z}^+$ we write $[\ell]$ to denote $\{1,\dots,\ell\}$.}

\paragraph{Running Times, and Bit Complexity.} Throughout this paper, we measure the running times of our algorithms in numbers of bit operations.
For a positive integer $n$, we denote by $\bit{n}$ its description complexity in binary, namely $\bit{n}  = \lceil \log_2 n \rceil$. Moreover, we represent a positive rational number $q$ as ${q_1 \over q_2}$, where $q_1$ and $q_2$ are relatively prime positive integers. The description complexity of $q$ is defined to be $\bit{q} = \bit{q_1}+\bit{q_2}$. We will assume that all $\epsilon$'s and $\delta$'s input to our algorithms are rational numbers.


\section{Learning a sum of
Bernoulli random variables from $\poly(1/\eps)$ samples} \label{sec:SOB}

In this section, we prove Theorem~\ref{thm:main} by providing a sample- and time-efficient algorithm for learning an unknown PBD $X=\sum_{i=1}^n X_i$. We start with an important ingredient in our analysis.

\medskip

\noindent {\bf A cover for PBDs.}
We make use of the following theorem, which provides a cover of the set ${\cal S} = {\cal S}_n$ of all PBDs of order-$n$. The theorem was given implicitly in~\cite{DP:oblivious11} and explicitly as Theorem~1 in~\cite{DP:cover}.

\begin{theorem} [Cover for PBDs] \label{thm: sparse cover theorem}
For all $\epsilon >0$, there exists an $\eps$-cover ${\cal S}_{\epsilon} \subseteq {\cal S}$ of ${\cal S}$ such that
\begin{enumerate}
\item
$|{\cal S}_{\epsilon}| \le n^2 + n \cdot \left({1 \over \epsilon}\right)^{O(\log^2{1/\epsilon})}$; and
\item
${\cal S}_{\epsilon}$ can be constructed in time linear in its representation size, i.e.,
${O}(n^2 \log n) + {O}(n \log n) \cdot \left({1 \over
\epsilon}\right)^{O(\log^2{1/\epsilon})}$.
\end{enumerate}
Moreover, if  $\{Y_i\}  \in {\cal S}_{\epsilon}$, then the collection of $n$ Bernoulli random variables $\{Y_i\}_{i=1,\dots,n}$ has one of the following forms, where
$k=k(\epsilon) \le C/\epsilon$ is a positive integer, for some absolute constant $C>0$:
\begin{itemize}
\item [(i)] ($k$-Sparse Form) There is some $ \ell \leq k^3=O(1/\epsilon^3)$
such that,
for all $i \leq \ell$, $\E[{Y_i}] \in \left\{{1 \over k^2}, {2\over k^2},\ldots, {k^2-1 \over k^2 }\right\}$ and,
for all $i >\ell $, $\E[{Y_i}] \in \{0,  1\}$.
\item [(ii)] ($k$-heavy Binomial Form) There is some $\ell \in \{1,\dots,n\}$
and $q \in \left\{ {1 \over n}, {2 \over n},\ldots, {n \over n}
\right\}$ such that,
for all $i \leq \ell$, $\E[{Y_i}] = q$ and,
for all $i >\ell$, $\E[{Y_i}] = 0$; moreover, $\ell,q$ satisfy
$\ell q \ge k^2$ and
$\ell q(1-q) \ge k^2- k-1.$
\end{itemize}

\noindent Finally, for every $\{X_i\}  \in {\cal S}$ for which there is no $\epsilon$-neighbor in ${\cal S}_{\epsilon}$ that is in sparse form, there exists some $\{Y_i\} \in {\cal S}_{\epsilon}$ in $k$-heavy Binomial form such that
\begin{itemize}
\item [(iii)] $\dtv(\sum_i X_i,\sum_i Y_i) \le \epsilon$; and
\item [(iv)] if $\mu = \E[\sum_i X_i]$, $\mu'= \E[\sum_i Y_i]$, $\sigma^2 = \Var[\sum_i X_i]$ and $\sigma'^2 = \Var[\sum_i Y_i]$, then
$
|\mu - \mu'| = O(1)$ and
$|\sigma^2-\sigma'^2| = O(1 + \epsilon \cdot (1+\sigma^2)).$
\end{itemize}
\end{theorem}

\noindent
We remark that the cover theorem as stated in~\cite{DP:cover} does not include the part of the above statement following ``finally.'' We provide a proof of this extension in Appendix~\ref{appendix:proof of extended cover theorem}.


\medskip

\noindent {\bf The Basic Learning Algorithm.} \blue{The high-level structure of our learning algorithms which give Theorem~\ref{thm:main} is provided in Algorithm {\tt Learn-PBD} of Figure~\ref{figure:learn-pBD}.  We instantiate this high-level structure, with appropriate technical modifications, in Section~\ref{sec:proofofmaintheorem}, where we give more detailed descriptions of the non-proper and proper algorithms that give parts (1) and (2) of Theorem~\ref{thm:main}.}

\begin{figure}[h!]
\framebox{
\medskip \noindent \begin{minipage}{16cm}

\medskip

{\tt Learn-PBD}$(n,\eps,\delta)$

\begin{enumerate}

 \item Run {\tt Learn-Sparse}$^X(n,\eps,\delta/3)$ to get hypothesis distribution $H_S$.

 \item Run {\tt Learn-Poisson}$^X(n,\eps,\delta/3)$ to get hypothesis distribution $H_P.$

\item Return the distribution which is the output of {\tt Choose-Hypothesis}$^X(H_S,H_P,\eps,\delta/3).$\\
\end{enumerate}

\end{minipage}}
\caption{{\tt Learn-PBD}$(n,\eps,\delta)$}
\label{figure:learn-pBD}
\end{figure}

\smallskip   At a high level, the subroutine {\tt Learn-Sparse} is given sample access to $X$ and is designed to find an $\eps$-accurate hypothesis $H_S$ with probability at least $1-\delta/3$, if the unknown PBD $X$ is $\eps$-close to some sparse form PBD inside the cover ${\cal S}_\epsilon$. Similarly, {\tt Learn-Poisson} is designed to find an $\eps$-accurate hypothesis $H_P$, if $X$ is not $\eps$-close to a sparse form PBD (in this case, Theorem~\ref{thm: sparse cover theorem} implies that $X$ must be $\eps$-close to some $k(\eps)$-heavy Binomial form PBD).  Finally, {\tt Choose-Hypothesis} is designed to choose one of the two
hypotheses $H_S,H_P$ as being $\eps$-close to $X.$  The following subsections specify these subroutines, as well as how the algorithm can be used to establish Theorem~\ref{thm:main}.
 We note that {\tt Learn-Sparse} and {\tt Learn-Poisson} do not return the distributions $H_S$ and $H_P$ as a list of probabilities for every point in $[n]$. They return instead a succinct description of these distributions in order to keep the running time of the algorithm logarithmic in $n$. Similarly, {\tt Choose-Hypothesis} operates with succinct descriptions of these distributions.

\subsection{Learning when $X$ is close to a sparse form PBD.}
\label{sec:learnclosetosparse}

Our starting point here is the simple observation that any PBD is a unimodal
distribution over the domain $\{0,1,\dots,n\}$.  (There is a simple inductive proof of this, or see Section~2 of \cite{KeilsonGerber:71}.)  This enables us to use the algorithm of Birg\'e \cite{Birge:97} for learning unimodal
distributions.  We recall Birg\'e's result, and refer the reader to Appendix~\ref{ap:birge} for an
explanation of how Theorem~\ref{thm:Birge unimodal} as stated below follows from \cite{Birge:97}.

\begin{theorem}[\cite{Birge:97}] \label{thm:Birge unimodal}
For all $n,\epsilon,\delta >0$, there is an algorithm that draws
$$O\left({\log n \over \epsilon^3} \log{1\over \delta} + {1 \over \epsilon^2} \log{1\over \delta} \log\log{1\over \delta} \right)$$
samples from an unknown unimodal distribution $X$ over $[n]$, does
$$\tilde{O}\left({\log^2 n \over \epsilon^3} \log^2{1\over \delta} \right)$$
bit-operations, and outputs a (succinct description of a) hypothesis distribution $H$ over $[n]$ that has the following form: $H$ is uniform over subintervals $[a_1,b_1],[a_2,b_2],\dots,[a_k,b_k]$, whose union $\cup_{i=1}^k [a_i,b_i]=[n]$, where $k=O\left({\log n \over \eps}\right).$ In particular, the algorithm outputs the lists $a_1$ through $a_k$ and $b_1$ through $b_k$, as well as the total probability mass that $H$ assigns to each subinterval $[a_i,b_i]$, $i=1,\ldots,k$. Finally, with probability  at least $1-\delta$, $\dtv(X,H) \leq \eps$.
\end{theorem}



The main result of this subsection
is the following:

\begin{lemma} \label{lem:learnsparse}
For all $n,\epsilon',\delta' >0$, there is an algorithm {\tt Learn-Sparse}$^X(n,\eps',\delta')$ that draws
$$O\left({1 \over \epsilon'^3}\log{1 \over \eps'}\log{1\over \delta'}+ {1 \over \epsilon'^2} \log{1\over \delta'} \log\log{1\over \delta'}\right)$$
samples from a target PBD $X$ over $[n]$, does
$$\log n \cdot \tilde{O}\left({1 \over \epsilon'^3}\log^2{1\over \delta'}\right)
$$
bit operations, and outputs a (succinct description of a) hypothesis distribution $H_S$ over $[n]$ that has the following form: its support is contained in an explicitly specified interval $[a,b] \subset [n]$, where $|b-a|=O(1/\eps'^3)$, and for every point in $[a,b]$ the algorithm explicitly specifies the probability assigned to that point by $H_S$.~\footnote{In particular, our algorithm will output a list of pointers, mapping every point in $[a,b]$ to some memory location where the probability assigned to that point by $H_S$ is written.} The algorithm has the following guarantee: if $X$ is $\eps'$-close to some sparse form PBD $Y$ in the cover ${\cal S}_{\epsilon'}$ of Theorem~\ref{thm: sparse cover theorem}, then with probability  at least $1-\delta'$, $\dtv(X,H_S) \leq c_1 \eps'$, for some absolute constant $c_1\ge 1$, and the support of $H_S$ lies in the support of $Y$.
\end{lemma}

The high-level idea of Lemma~\ref{lem:learnsparse} is quite simple.  We truncate $O(\eps')$ of the probability mass
from each end of $X$ to obtain a conditional distribution $X_{[\hat{a},\hat{b}]}$; since $X$ is unimodal so is $X_{[\hat{a},\hat{b}]}$.  If $\hat{b}-\hat{a}$ is larger than $ O(1/\eps'^3)$ then the algorithm outputs ``fail''
(and $X$ could not have been close to a sparse-form distribution in the cover).  Otherwise, we use Birg\'{e}'s algorithm to learn the unimodal distribution $X_{[\hat{a},\hat{b}]}$.
\blue{A detailed description of the algorithm is given in Figure~\ref{figure:learnsparse} below.}

\begin{figure}[h!]
\framebox{
\medskip \noindent \begin{minipage}{16cm}

\blue{

\medskip

{\tt Learn-Sparse}$^X(n,\eps',\delta')$

\begin{enumerate}

\item Draw $M=32\log(8/\delta')/\eps'^2$ samples from $X$ and sort them to obtain a list of values $0 \leq s_1 \leq \cdots \leq s_M \leq n.$

\item Define $\hat{a}:=s_{\lceil 2 \eps' M \rceil}$ and $\hat{b}:=s_{\lfloor(1 - 2 \eps')M\rfloor}$.

\item If $\hat{b}-\hat{a} > (C/\eps')^3$ (where $C$ is the constant in the statement of Theorem~\ref{thm: sparse cover theorem}), output ``fail'' and return the (trivial)
  hypothesis which puts probability mass $1$ on the point $0$.

\item  Otherwise, run Birg\'e's unimodal distribution learner (Theorem~\ref{thm:Birge unimodal}) on the conditional distribution $X_{[\hat{a},\hat{b}]}$ and output the hypothesis that it returns.

\end{enumerate}
}

\end{minipage}}
\caption{{\tt Learn-Sparse}$^X(n,\eps',\delta')$ }
\label{figure:learnsparse}
\end{figure}

\medskip

\begin{prevproof}{Lemma}{lem:learnsparse}
\blue{As described in Figure~\ref{figure:learnsparse}, algorithm {\tt Learn-Sparse}$^X(n,\eps',\delta')$}  first draws $M=32\log(8/\delta')/\eps'^2$ samples from $X$ and sorts them to obtain a list of values $0 \leq s_1 \leq \cdots \leq s_M \leq n.$
We claim the following \blue{about the values $\hat{a}$ and $\hat{b}$ defined in Step~2
of the algorithm:}
\begin{claim} \label{claim:anchors are fine}
With probability at least $1-\delta'/2$, we have $X(\leq \hat{a}) \in [3\eps'/2,5\eps'/2]$ and
$X(\leq \hat{b}) \in [1-5\eps'/2,1-3\eps'/2]$.
\end{claim}
\begin{proof}
We only show that $X(\le \hat{a}) \ge 3\eps'/2$ with probability at least $1- \delta'/8$, since the arguments for $X(\le \hat{a}) \le 5\eps'/2$, $X(\le \hat{b}) \le 1-3\eps'/2$ and $X(\le \hat{b}) \ge 1-5\eps'/2$ are identical. Given that each of these conditions is met with probability at least $1-\delta'/8$, the union bound establishes our claim.

To show that $X(\le \hat{a}) \ge 3\eps'/2$ is satisfied with probability at least $1- \delta'/8$ we argue
as follows: Let  $\alpha'=\max\{ i~|~X(\le i) < 3\eps'/2\}$. Clearly, $X(\le \alpha') < 3\eps'/2$ while $X(\le \alpha'+1)\ge 3\eps'/2$. Given this, if $M$
samples are drawn from $X$ then the expected number of them that are
$\le \alpha'$ is at most $3\eps' M/2$. It follows then from the Chernoff bound that the probability that more than ${7 \over 4}\eps' M$ samples are $\le \alpha'$ is at most $e^{-(\eps'/4)^2 M/2} \le \delta'/8$. Hence
except with this failure probability, we have
$\hat{a} \ge \alpha'+1$, which implies that $X(\le \hat{a}) \ge 3\eps'/2$.
\end{proof}

\blue{As specified in Steps~3 and~4,} if $\hat{b}-\hat{a} > (C/\eps')^3$, where $C$ is the constant in the statement of Theorem~\ref{thm: sparse cover theorem},  the algorithm outputs ``fail'', returning the trivial hypothesis which puts probability mass $1$ on the point $0$. Otherwise, the algorithm runs Birg\'e's unimodal distribution learner (Theorem~\ref{thm:Birge unimodal}) on the conditional distribution
$X_{[\hat{a},\hat{b}]}$, and outputs the result of Birg\'e's algorithm. Since $X$ is unimodal, it follows that $X_{[\hat{a},\hat{b}]}$ is also unimodal, hence Birg\'e's algorithm is appropriate for learning it. The way we apply Birg\'e's algorithm to learn $X_{[\hat{a},\hat{b}]}$ given samples from the original distribution $X$ is the obvious one: we draw samples from $X$, ignoring all samples that fall outside of $[\hat{a},\hat{b}]$, until the right $O(\log(1/\delta') \log (1/\eps')/\eps'^3)$ number of samples fall inside $[\hat{a},\hat{b}]$, as required by Birg\'e's algorithm for learning a distribution of support of size $(C/\eps')^3$ with probability at least $1-\delta'/4$. Once we have the right number of samples in $[\hat{a},\hat{b}]$, we  run Birg\'e's algorithm to learn the conditional distribution $X_{[\hat{a},\hat{b}]}$. Note that the number of samples we need to draw from $X$ until  the right $O(\log(1/\delta') \log (1/\eps')/\eps'^3)$ number of samples fall inside $[\hat{a},\hat{b}]$ is still $O(\log(1/\delta') \log (1/\eps')/\eps'^3)$, with probability at least $1-\delta'/4$. Indeed, since $X([\hat{a},\hat{b}])=1-O(\eps')$, it follows from the Chernoff bound that with probability at least $1-\delta'/4$, if $K=\Theta(\log(1/\delta') \log (1/\eps')/\eps'^3)$ samples are drawn from $X$, at least $K(1-O(\eps'))$ fall inside $[\hat{a},\hat{b}]$.

{\bf Analysis:} It is easy to see that the sample complexity of our algorithm is as promised. For the running time, notice that, if Birg\'e's algorithm is invoked, it will return two lists of numbers $a_1$ through $a_k$ and $b_1$ through $b_k$, as well as a list of probability masses $q_1,\ldots,q_k$ assigned to each subinterval $[a_i,b_i]$, $i=1,\ldots,k$, by the hypothesis distribution $H_S$, where $k=O(\log(1/\eps')/\eps')$. In linear time, we can compute a list of probabilities $\hat{q}_1,\ldots,\hat{q}_k$, representing the probability assigned by $H_S$ to every point of subinterval $[a_i,b_i]$, for $i=1,\ldots,k$. So we can represent our output hypothesis $H_S$ via a data structure that maintains $O(1/\eps'^3)$ pointers, having one pointer per point inside $[a,b]$. The pointers map points to probabilities assigned by $H_S$ to these points. Thus turning the output of Birg\'e's algorithm into an explicit distribution over $[a,b]$ incurs linear overhead in our running time, and hence the running time of our algorithm is also as promised. (See Appendix~\ref{ap:birge} for an explanation of the running time of Birg\'e's algorithm.)  Moreover, we also note that the output distribution has the promised structure, since in one case it has a single atom at $0$ and in the other case it is the output of Birg\'e's algorithm on a distribution of support of size $(C/\eps')^3$.

It only remains to justify the last part of the lemma. Let $Y$ be the sparse-form PBD that $X$ is close to; say that $Y$ is supported on
$\{a',\dots,b'\}$ where $b'-a' \leq (C/\eps')^3.$  Since $X$ is $\eps'$-close to $Y$ in total variation distance it must be the case
that $X(\leq a'-1) \leq \eps'$. \ignore{\rnote{This part had proceeded a little differently, using unimodality to conclude that $X(\leq a'-2) < \eps'$ -- but it seemed to me $X$ could be unimodal and still not put weight on $a'-1$.  So I changed Claim 2 earlier to use $3\eps/2$ instead of $\eps$ and now I think it's ok.}}Since $X(\leq \hat{a}) \geq 3\eps'/2$ by Claim~\ref{claim:anchors are fine}, it must be the case that $\hat{a} \geq a'$.  Similar
arguments give that $\hat{b} \leq b'$.  So the interval $[\hat{a},\hat{b}]$ is contained in $[a',b']$ and has length at most $(C/\eps')^3$.  This means that Birg\'e's algorithm is indeed used correctly by our algorithm to learn $X_{[\hat{a},\hat{b}]}$, with probability at least $1-\delta'/2$ (that is, unless Claim~\ref{claim:anchors are fine} fails).  Now it follows from the correctness of Birg\'e's algorithm (Theorem~\ref{thm:Birge unimodal}) and the discussion above, that the hypothesis $H_S$ output when Birg\'e's algorithm is invoked satisfies $\dtv(H_S,X_{[\hat{a},\hat{b}]})\le \epsilon',$ with probability at least $1-\delta'/2$, i.e., unless either Birg\'e's algorithm fails, or we fail to get the right number of samples landing inside $[\hat{a},\hat{b}]$.
To conclude the proof of the lemma  we note that:
\begin{eqnarray*}
2 \dtv(X,X_{[\hat{a},\hat{b}]}) &=&
\sum_{i \in [\hat{a},\hat{b}]}|X_{[\hat{a},\hat{b}]}(i)-X(i)| +  \sum_{i \notin [\hat{a},\hat{b}]}|X_{[\hat{a},\hat{b}]}(i)-X(i)|\\
&=& \sum_{i \in [\hat{a},\hat{b}]} \Big| {1 \over  X({[\hat{a},\hat{b}]})}X(i)-X(i) \Big| +  \sum_{i \notin [\hat{a},\hat{b}]}X(i)\\
&=& \sum_{i \in [\hat{a},\hat{b}]} \Big| {1 \over  1-O(\eps')}X(i)-X(i) \Big| +  O(\eps')\\
&=& {O(\eps') \over 1-O(\eps')}\sum_{i \in [\hat{a},\hat{b}]} \Big| X(i) \Big| +  O(\eps')\\
& =& O(\eps').
\end{eqnarray*}
So the triangle inequality gives: $\dtv(H_S,X) = O(\eps')$, and Lemma~\ref{lem:learnsparse}
is proved. \end{prevproof}

\subsection{Learning when $X$ is close to a $k$-heavy Binomial Form PBD.} \label{sec:kheavy}

\begin{lemma} \label{lem:learnbinomial}
For all $n,\epsilon',\delta' >0$, there is an algorithm {\tt Learn-Poisson}$^X(n,\eps',\delta')$ that  draws
$$O(\log(1/\delta')/\eps'^2)$$
samples from a target PBD $X$ over $[n]$, does
$$O(\log n\cdot \log (1/\delta')/\eps'^2)$$
bit operations, and returns two parameters $\hat{\mu}$ and $\hat{\sigma}^2$. The algorithm has the following guarantee: Suppose $X$ is not $\eps'$-close to any sparse form PBD in the cover ${\cal S}_{\epsilon'}$ of Theorem~\ref{thm: sparse cover theorem}.
Let $H_P=TP(\hat{\mu},\hat{\sigma}^2)$
be the translated Poisson distribution with parameters $\hat{\mu}$ and $\hat{\sigma}^2$.
Then with probability at least $1-\delta'$ we have $\dtv(X,H_P) \leq
c_2 \eps'$ for some absolute constant $c_2 \ge 1$.
 \end{lemma}

Our proof plan is to exploit the structure of the cover of Theorem~\ref{thm: sparse cover theorem}. In particular, if $X$ is not $\eps'$-close to any sparse form PBD in the cover, it must be $\epsilon'$-close to a PBD in heavy Binomial form with approximately the same mean and variance as $X$, as specified by the final part of the cover theorem. Hence, a natural strategy is to obtain estimates $\hat{\mu}$ and $\hat{\sigma}^2$ of the mean and variance of the unknown PBD $X$, and output as a hypothesis a translated Poisson distribution with parameters $\hat{\mu}$ and $\hat{\sigma}^2$. We show that this strategy is a successful one.  {Before providing the details,} we {highlight} two facts
{that we will establish in the subsequent analysis and}
that will be used later.  The first is that, assuming $X$ is not $\eps'$-close to any sparse form PBD in the cover ${\cal S}_{\epsilon'}$, its variance $\sigma^2$ satisfies
\begin{equation} \label{eq:sigmabig}
\sigma^2 = \Omega(1/\eps'^2) \ge \theta^2 \quad \text{for some universal constant~}\theta.
\end{equation}
The second is that under the same assumption, the
estimates $\hat{\mu}$ and $\hat{\sigma}^2$ of the mean $\mu$ and variance $\sigma^2$ of $X$ that we obtain satisfy the following bounds with probability at least $1-\delta$:
\begin{equation} \label{eq:goodparams}
|\mu-\hat{\mu}| \le \epsilon' \cdot \sigma \quad\text{and}\quad |\sigma^2 - \hat{\sigma}^2| \le \epsilon' \cdot \sigma^2.
\end{equation}

\blue{See Figure~\ref{figure:learnpoisson} and the associated Figure~\ref{figure:A}
for a detailed description of the {\tt Learn-Poisson}$^X(n,\eps',\delta')$ algorithm.}

\begin{figure}[h!]
\framebox{
\medskip \noindent \begin{minipage}{16cm}

\blue{

\medskip

{\tt Learn-Poisson}$^X(n,\eps',\delta')$

\begin{enumerate}

\item  Let $\epsilon=\epsilon'/ \sqrt{4+{1 \over \theta^2}}$ and $\delta = \delta'.$

\item Run algorithm ${\cal A}(n,\eps,\delta)$ to obtain an estimate $\hat{\mu}$ of $\E[X]$ and an estimate $\hat{\sigma}^2$ of $\Var[X]$.

\item  Output the translated Poisson distribution $TP(\hat{\mu},\hat{\sigma}^2)$.

\end{enumerate}

}

\end{minipage}}
\caption{{\tt Learn-Poisson}$^X(n,\eps',\delta')$.  The value $\theta$ used in Line~1 is
the universal constant specified in the proof of Lemma~\ref{lem:learnbinomial}.}
\label{figure:learnpoisson}
\end{figure}

\begin{figure}[h!]
\framebox{
\medskip \noindent \begin{minipage}{16cm}

\blue{

\medskip

${\cal A}(n,\eps,\delta)$

\begin{enumerate}

\item Let $r=O(\log 1 /\delta).$ For $i=1,\dots,r$ repeat the following:

\begin{enumerate}

\item Draw $m = \lceil 3 / \epsilon^2 \rceil$ independent samples $Z_{i,1},\dots,Z_{i,m}$ from $X$.

\item Let $\hat{\mu}_i = {\frac {\sum_j Z_{i,j}} m},$ $\hat{\sigma}^2_i={\sum_j (Z_{i,j} - {1 \over m}\sum_k Z_{i,k})^2 \over m-1}$.

\end{enumerate}

\item  Set $\hat{\mu}$ to be the median of $\hat{\mu}_1,\dots,\hat{\mu}_r$ and
set $\hat{\sigma}^2$ to be the median of $\hat{\sigma}^2_1,\dots,\hat{\sigma}^2_r$.

\item  Output $\hat{\mu}$ and $\hat{\sigma}^2$.

\end{enumerate}
}

\end{minipage}}
\caption{${\cal A}(n,\eps,\delta)$ }
\label{figure:A}
\end{figure}

\medskip

\begin{prevproof}{Lemma}{lem:learnbinomial}
We start by showing that we can estimate the mean and variance of the target PBD $X$.

\begin{lemma}\label{lem:can estimate mean and variance}
For all $n, \epsilon,\delta>0$,  there exists an algorithm ${\cal A}(n,\epsilon,\delta)$ with the following properties:  given access to a PBD $X$ of order $n$, it produces estimates $\hat{\mu}$ and $\hat{\sigma}^2$ for $\mu=\E[X]$ and $\sigma^2=\Var[X]$ respectively such that with probability at least $1-\delta$:
\[|\mu-\hat{\mu}| \le \epsilon \cdot \sigma \text{~~~~~~~ and ~~~~~~~}
|\sigma^2 - \hat{\sigma}^2| \le \epsilon \cdot \sigma^2 \sqrt{4+{1 \over \sigma^2}}.
\]
The algorithm uses $$O(\log(1/\delta)/\epsilon^2)$$ samples and runs in time
$$O(\log n \log(1/\delta)/\epsilon^2).$$
\end{lemma}
\begin{proof}
We treat the estimation of $\mu$ and $\sigma^2$ separately. For both estimation problems we show how to use $O(1/\epsilon^2)$ samples to obtain estimates $\hat{\mu}$ and $\hat{\sigma}^2$ achieving the required guarantees with probability at least $2/3$ (we refer to these as ``weak estimators''). Then a routine procedure allows us to boost the success probability to $1-\delta$ at the expense of a multiplicative factor $O(\log 1/\delta)$ on the number of samples. While we omit the details of the routine boosting argument, we remind the reader that it involves running the weak estimator $O(\log 1/\delta)$ times to obtain estimates $\hat{\mu}_1,\ldots,\hat{\mu}_{O(\log{1/\delta})}$ and outputting the median of these estimates, and similarly for estimating $\sigma^2$.

We proceed to specify and analyze the weak estimators for $\mu$ and $\sigma^2$ separately:
\begin{itemize}
\item {\em Weak estimator for $\mu$:} Let $Z_1,\ldots, Z_m$ be independent samples from $X$, and let $\hat{\mu}={\sum_i Z_i \over m}$. Then $$\E[\hat{\mu}] = \mu~~\text{and}~~\Var[\hat{\mu}]= {1 \over m} \Var[X] = {1 \over m} \sigma^2.$$
So Chebyshev's inequality implies that $$\Pr[|\hat{\mu} - \mu| \ge t \sigma/\sqrt{m}] \le {1 \over t^2}.$$
Choosing $t = \sqrt{3}$ and $m = \lceil 3 / \epsilon^2 \rceil$, the above imply that $|\hat{\mu}-\mu| \le \epsilon \sigma$ with probability at least $2/3$.
\item {\em Weak estimator for $\sigma^2$:} Let $Z_1,\ldots, Z_m$ be independent samples from $X$, and let $\hat{\sigma}^2={\sum_i (Z_i - {1 \over m}\sum_i Z_i)^2 \over m-1}$ be the
     unbiased sample variance. (Note the use of Bessel's correction.) Then it can be checked \cite{Johnson:2003} that
     $$\E[\hat{\sigma}^2] = \sigma^2~~\text{and}~~\Var[\hat{\sigma}^2]= \sigma^4 \left( {2 \over m-1} + {\kappa \over m} \right),$$
where $\kappa$ is the \blue{excess kurtosis} of the distribution of $X$ \blue{(i.e. $\kappa = {\frac {\E[(X-\mu)^4]}{\sigma^4}} - 3$)}. To bound $\kappa$ in terms of $\sigma^2$ suppose that $X= \sum_{i=1}^n X_i$, where $\E[X_i]=p_i$ for all $i$. Then
\begin{align*}\kappa&={1 \over \sigma^4} \sum_i (1-6 p_i(1-p_i))(1-p_i)p_i \quad \quad \quad \text{(see \cite{LJ:2005})}\\
\ignore{ &\le{1 \over \sigma^4} \sum_i |1-6 p_i(1-p_i)|(1-p_i)p_i\\}
&\le{1 \over \sigma^4} \sum_i (1-p_i)p_i = {1 \over \sigma^2}.
\end{align*}
Hence, $\Var[\hat{\sigma}^2]= \sigma^4 \left( {2 \over m-1} + {\kappa \over m} \right) \le {\sigma^4 \over m} (4+{1 \over \sigma^2}).$
So Chebyshev's inequality implies that $$\Pr\left[|\hat{\sigma}^2 - \sigma^2| \ge t {\sigma^2 \over \sqrt{m}}\sqrt{4+{1\over \sigma^2}} \right] \le {1 \over t^2}.$$
Choosing $t = \sqrt{3}$ and $m = \lceil 3 / \epsilon^2 \rceil$, the above imply that $|\hat{\sigma}^2-\sigma^2| \le \epsilon \sigma^2 \sqrt{4+{1\over \sigma^2}}$ with probability at least $2/3$.
\end{itemize}
\end{proof}

We proceed to prove Lemma~\ref{lem:learnbinomial}. {\tt Learn-Poisson}$^X(n,\eps',\delta')$ runs ${\cal A}(n,\epsilon,\delta)$ from Lemma~\ref{lem:can estimate mean and variance} with appropriately chosen $\epsilon=\epsilon(\epsilon')$ and $\delta=\delta(\delta')$, given below, and then outputs the translated Poisson distribution
$TP(\hat{\mu}, \hat{\sigma}^2)$, where $\hat{\mu}$ and $\hat{\sigma}^2$ are the estimated mean and variance of $X$ output by ${\cal A}$. Next, we show how to choose $\epsilon$ and $\delta$, as well as why the desired guarantees are satisfied by the output distribution.

If $X$ is not $\eps'$-close to any PBD in sparse form inside the cover ${\cal S}_{\eps'}$ of Theorem~\ref{thm: sparse cover theorem}, there exists a PBD $Z$ in $(k=O(1/\eps'))$-heavy Binomial form inside ${\cal S}_{\eps'}$ that is within total variation distance $\eps'$ from $X$.  We use the existence of such $Z$ to obtain lower bounds on the mean and variance of $X$. Indeed, suppose that the distribution of $Z$ is $\mathrm{Bin}(\ell,q)$, a Binomial with parameters $\ell,q$. Then Theorem~\ref{thm: sparse cover theorem} certifies  that the following conditions are satisfied by the parameters $\ell, q$, $\mu = \E[X]$ and $\sigma^2=\Var[X]$:
\begin{itemize}
\item [(a)] $\ell q \ge k^2$;
\item [(b)] $\ell q(1-q) \ge k^2- k-1$;
\item [(c)] $|\ell q - \mu| = O(1)$; and
\item [(d)] $|\ell q (1-q)-\sigma^2| = O(1 + \epsilon' \cdot (1+\sigma^2))$.
\end{itemize}
In particular, conditions (b) and (d) above imply that
\begin{equation*}
\sigma^2 = \Omega( k^2) = \Omega(1/\eps'^2) \ge \theta^2,
\end{equation*} for some universal constant $\theta$, establishing~(\ref{eq:sigmabig}).
In terms of this $\theta$, we choose $\epsilon=\epsilon'/ \sqrt{4+{1 \over \theta^2}}$ and $\delta=\delta'$ for the application of Lemma~\ref{lem:can estimate mean and variance} to obtain---from $O(\log(1/\delta')/\epsilon'^2)$ samples---estimates $\hat{\mu}$ and $\hat{\sigma}^2$ of $\mu$ and $\sigma^2$.

From our choice of parameters and the guarantees of Lemma~\ref{lem:can estimate mean and variance}, it follows that, if $X$ is not $\eps'$-close to any PBD in sparse form inside the cover  ${\cal S}_{\eps'}$, then  with probability at least $1-\delta'$ the estimates  $\hat{\mu}$ and $\hat{\sigma}^2$  satisfy:
\begin{equation*}
|\mu-\hat{\mu}| \le \epsilon' \cdot \sigma \quad\text{and}\quad |\sigma^2 - \hat{\sigma}^2| \le \epsilon' \cdot \sigma^2,
\end{equation*}
establishing~(\ref{eq:goodparams}). Moreover, if $Y$ is a random variable distributed according to the translated Poisson distribution $TP(\hat{\mu}, \hat{\sigma}^2)$, we show that $X$ and $Y$ are within $O(\epsilon')$ in total variation distance, concluding the proof of Lemma~\ref{lem:learnbinomial}.

\begin{claim} \label{claim:final claim heavy case}
If $X$ and $Y$ are as above, then $\dtv(X,Y) \leq O(\eps')$.
\end{claim}
\begin{proof}
We make use of Lemma~\ref{lem:translated Poisson approximation}.
Suppose that $X = \sum_{i=1}^n X_i$, where $\E[X_i]=p_i$ for all $i$. Lemma~\ref{lem:translated Poisson approximation} implies that
\begin{eqnarray}
 \dtv(X, TP(\mu, \sigma^2)) &\le&
\frac{\sqrt{\sum_{i}p_i^3(1-p_i)}+2}{\sum_{i}p_i(1-p_i)} \nonumber \\
&\le& \frac{\sqrt{\sum_{i}p_i(1-p_i)}+2}{\sum_{i}p_i(1-p_i)} \nonumber\\
&\le& \frac{1}{\sqrt{\sum_{i}p_i(1-p_i)}}+\frac{2}{\sum_{i}p_i(1-p_i)}
\nonumber \\
&=& \frac{1}{\sigma}+\frac{2}{\sigma^2}
\nonumber\\
& = &O(\epsilon').
\label{eq: tv triangle inequality stage 2a}
\end{eqnarray}
It remains to bound the total variation distance between the translated Poisson distributions $TP(\mu,\sigma^2)$ and $TP(\hat{\mu},\hat{\sigma}^2)$. For this we use Lemma~\ref{lem: variation distance between translated Poisson distributions}.  Lemma~\ref{lem: variation distance between translated Poisson distributions} implies
\begin{eqnarray}
\dtv(TP(\mu, \sigma^2), TP(\hat{\mu}, \hat{\sigma}^2)) &\le& \frac{|\mu-\hat{\mu}|}{\min(\sigma,\hat{\sigma})}+ \frac{|\sigma^2-\hat{\sigma}^2|+1}{\min(\sigma^2,\hat{\sigma}^2)} \notag\\
&\le& \frac{\epsilon' \sigma}{\min(\sigma,\hat{\sigma})}+ \frac{\epsilon' \cdot \sigma^2+1}{\min(\sigma^2,\hat{\sigma}^2)} \notag\\
&\le& \frac{\epsilon' \sigma}{\sigma/\sqrt{1-\epsilon'}}+ \frac{\epsilon' \cdot \sigma^2+1}{\sigma^2/(1-\epsilon')} \notag\\
&=&O(\epsilon') + \frac{O(1-\epsilon')}{\sigma^2}\notag\\
&=&O(\epsilon') + O(\epsilon'^2) \notag\\
&=&O(\epsilon'). \label{eq:tv between the two translated Poisson's}
\end{eqnarray}
The claim follows from \eqref{eq: tv triangle inequality stage 2a}, \eqref{eq:tv between the two translated Poisson's} and the triangle inequality.
\end{proof}
The proof of Lemma~\ref{lem:learnbinomial} is concluded. We remark that the algorithm described above does not need to know a priori whether or not $X$ is $\eps'$-close to a PBD in sparse form inside the cover ${\cal S}_{\eps'}$ of Theorem~\ref{thm: sparse cover theorem}. The algorithm simply runs the estimator of Lemma~\ref{lem:can estimate mean and variance} with $\epsilon = \epsilon'/ \sqrt{4+{1 \over \theta^2}}$ and $\delta'=\delta$ and outputs whatever estimates $\hat{\mu}$ and $\hat{\sigma}^2$ the algorithm of Lemma~\ref{lem:can estimate mean and variance} produces.
\end{prevproof}

\subsection{ Hypothesis testing.} \label{sec:choosehypothesis}

Our hypothesis testing routine {\tt Choose-Hypothesis}$^X$ uses samples from the unknown distribution $X$ to run a ``competition'' between two candidate hypothesis distributions $H_1$ and $H_2$  over $[n]$ that are given in the input.  We show that if at least one of the two candidate hypotheses is close to the unknown distribution $X$, then with high probability over the samples drawn from $X$ the routine selects as winner a candidate that is close to $X$.
  This basic approach of running a competition between candidate hypotheses is quite similar to the ``Scheff\'e estimate'' proposed by Devroye and Lugosi (see \cite{DL96,DL97}
and Chapter~6 of \cite{DL:01}, as well as \cite{Yatracos85}), but our notion of competition here is different.

We obtain the following lemma, postponing all running-time analysis to the next section.
\begin{lemma} \label{lem:choosehypothesis}
There is an algorithm {\tt Choose-Hypothesis}$^X({H}_1,{H}_2,\eps',\delta')$ which
is given sample access to distribution $X$, two hypothesis distributions $H_1,H_2$ for $X$,  an accuracy
parameter $\eps'>0$, and a confidence parameter $\delta'>0.$  It makes
$$m=O(\log(1/\delta')/\eps'^2)$$
draws from $X$ and returns some $H \in \{H_1,H_2\}.$  If $\dtv(H_i,X) \leq \eps'$ for some $i \in \{1,2\}$, then with probability at least $1-\delta'$ the distribution $H$ that {\tt
Choose-Hypothesis} returns has $\dtv(H,X) \leq 6 \eps'.$
\end{lemma}



\begin{prevproof}{Lemma}{lem:choosehypothesis}
Figure~\ref{figure:Choose-Hypothesis} describes how the competition between $H_1$ and $H_2$ is carried out.

\begin{figure}[h!]
\framebox{
\medskip \noindent \begin{minipage}{16cm}

\medskip

{\tt Choose-Hypothesis}$(H_1,H_2,\eps',\delta')$\\
{\sc Input:} Sample access to distribution $X$; a pair of hypothesis distributions $(H_1, H_2)$; $\epsilon',\delta'>0$.\\

Let ${\cal W}$ be the support of $X$, ${\cal W}_1={\cal W}_1(H_1,H_2) := \left\{w \in \mathcal{W}~\vline~H_1(w) > H_2(w) \right\}$, and $p_1 = H_1({\cal W}_1)$, $p_2 = H_2({\cal W}_1)$. {\em /* Clearly, $p_1 > p_2$ and $\dtv(H_1, H_2) = p_1-p_2$. */}
\begin{enumerate}
\item
If $p_1-p_2\leq 5 \eps'$, declare a draw and return either $H_i$. Otherwise:

\item
Draw $m=2 {\log(1/\delta') \over \eps'^2}$ samples $s_1,\ldots,s_m$ from $X$, and let $\tau = {1 \over m} | \{i~|~s_i
\in {\cal W}_1 \}|$ be the fraction of  samples that fall inside ${\cal W}_1.$

\item
If $\tau > p_1- {3 \over 2} \eps'$, declare $H_1$ as winner and return $H_1$; otherwise,

\item
if $\tau < p_2+ {3 \over 2} \eps'$, declare $H_2$ as winner and return $H_2$; otherwise,

\item
declare a draw and return either $H_i$.\\
\end{enumerate}

\end{minipage}}
\caption{{\tt Choose-Hypothesis$(H_1,H_2,\epsilon',\delta')$}}
\label{figure:Choose-Hypothesis}
\end{figure}

The correctness of {\tt Choose-Hypothesis} is an immediate consequence of the following claim.  (In fact for
 Lemma~\ref{lem:choosehypothesis} we only need item (i) below, but item (ii) will be handy later
 in the proof of Lemma~\ref{lem:log-cover-size}.)
\begin{claim} \label{lem:kostas3}
Suppose that $\dtv(X,H_i) \leq \eps'$, for {some} $i \in \{1,2\}$. Then:
\begin{itemize}
\item[(i)] if $\dtv(X, H_{3-i})>6 \eps'$, the probability that {\tt Choose-Hypothesis}$^X({H}_1,{H}_2,\eps',\delta')$ does not declare $H_i$ as the winner is at most $2e^{- { m \eps'^2/2 } }$, where $m$ is chosen as in the description of the algorithm.  (Intuitively,
if $H_{3-i}$ is very bad then it is very likely that $H_i$ will be declared winner.)

\item[(ii)]  if $\dtv(X, H_{3-i})>4 \eps'$, the probability that {\tt Choose-Hypothesis}$^X({H}_1,{H}_2,\eps',\delta')$ declares $H_{3-i}$ as the winner is at most $2e^{- { m \eps'^2 /2} }$.  (Intuitively, if $H_{3-i}$
is only moderately bad then a draw is possible but it is very unlikely that $H_{3-i}$ will be declared winner.)
\end{itemize}
\end{claim}
\begin{proof}
Let $r=X({\cal W}_1)$. The definition of the total variation distance implies that $|r-p_i| \le \eps'$. Let us define independent indicators $\{Z_j\}_{j=1}^m$ such that, for all $j$, $Z_j=1$ iff $s_j \in {\cal W}_1$. Clearly, $\tau={1 \over m} \sum_{j=1}^m Z_j$ and $\mathbb{E}[\tau]=\mathbb{E}[Z_j]=r$. Since the $Z_j$'s are mutually independent, it follows from the Chernoff bound that $\Pr[|\tau - r| \ge {\eps'/2}] \le 2 e^{- {m \eps'^2/2 } }$. Using $|r-p_i| \le \eps'$ we get that $\Pr[|\tau -
p_i| \ge {3\eps'/2}] \le 2 e^{- {m \eps'^2/2 } }$. Hence:
\begin{itemize}

\item For part (i):  If $\dtv(X, H_{3-i}) > 6 \eps'$, from the triangle inequality we get that $p_1-p_2=\dtv(H_1, H_2) > 5 \eps' $.  Hence, the algorithm will go beyond step 1, and with probability at least $1-2e^{- {m \eps'^2/2 } }$, it will stop at step 3 (when $i=1$) or step 4 (when $i=2$), declaring $H_i$ as the winner of the competition between $H_1$
and $H_2$.

\item For part (ii):  If $p_1-p_2 \leq 5 \eps'$ then the competition declares a draw, hence $H_{3-i}$ is not the winner. Otherwise we have $p_1 - p_2 > 5\eps'$ and the above arguments imply that the competition between $H_1$ and $H_2$ will declare $H_{3-i}$ as the winner with probability at most $2e^{- {m \eps'^2/2 } }$.
\end{itemize}
This concludes the proof of Claim~\ref{lem:kostas3}.
\end{proof}In view of Claim~\ref{lem:kostas3}, the proof of Lemma~\ref{lem:choosehypothesis} is concluded.\end{prevproof}

Our {\tt Choose-Hypothesis} algorithm implies a generic learning algorithm of independent interest.

\begin{lemma} \label{lem:log-cover-size}
Let ${\cal S}$ be an arbitrary set of distributions over a finite domain.  Moreover, let ${\cal S}_\eps \subseteq {\cal S}$ be an $\eps$-cover of ${\cal S}$ of size $N$, for some $\epsilon>0$.\ignore{ which can be enumerated in time $\poly(|{\cal S}_\eps|)$.} For all $\delta>0$, there is an algorithm that uses $$O(\eps^{-2}\log N \log(1/\delta))$$ samples from
an unknown distribution $X \in {\cal S}$ and,
with probability at least $1-\delta$, outputs a distribution $Z \in {\cal S}_{\eps}$ that satisfies
$\dtv(X,Z) \leq 6\eps.$
\end{lemma}

\begin{proof}
The algorithm performs a tournament, by running
{\tt Choose-Hypothesis}$^X(H_i,H_j,\eps,$ $\delta/(4N))$ for every pair
$(H_i,H_j)$, $i < j$, of distributions in ${\cal S}_\eps$.  Then it
outputs any distribution $Y^\star \in {\cal S}_\eps$ that was never a loser (i.e., won or tied
against all other distributions in the cover).  If no such distribution exists in ${\cal S}_\eps$ then the algorithm says
``failure,'' and outputs an arbitrary distribution from ${\cal S}_\eps$.

Since ${\cal S}_\eps$ is an $\eps$-cover of $\mathcal{S}$, there exists some $Y \in {\cal S}_\eps$ such that
$\dtv(X,Y) \leq \eps.$ We first argue that with high probability this distribution $Y$ never
loses a competition against any other $Y' \in {\cal S}_\eps$ (so the algorithm does not output ``failure'').  Consider any $Y'
\in {\cal S}_\eps$. If $\dtv(X, Y') > 4 \eps$, by Claim~\ref{lem:kostas3}(ii) the probability that $Y$ loses to
$Y'$ is at most $2e^{-m \eps^2 /2} {\leq}
{\delta \over 2N}.$ On the other hand, if $\dtv(X, Y') \leq 4 \eps$, the triangle
inequality gives that $\dtv(Y,Y') \leq 5 \eps$ and thus $Y$ draws against $Y'.$  A union
bound over all $N-1$ distributions in {${\cal S}_\eps\setminus \{ Y\}$} shows that with probability at least $1-\delta/2$, the distribution $Y$ never loses
a competition.

We next argue that with probability at least $1-\delta/2$, every distribution $Y' \in {\cal S}_\eps$ that never loses must be
close to $X.$ Fix a distribution $Y'$ such that $\dtv(Y', X) > 6 \eps$. Lemma~\ref{lem:kostas3}(i) implies that
$Y'$ loses to $Y$ with probability at least $1 - 2e^{-m \eps^2 /2} \geq 1 - \delta/(2N)$.  A union bound gives that with
probability at least $1-\delta/2$, every distribution $Y'$ that has $\dtv(Y', X) > 6 \eps$ loses some competition.

Thus, with overall probability at least $1-\delta$, the tournament does not output ``failure'' and outputs some distribution $Y^\star$ such that $\dtv(X, Y^\star) \le 6 \eps.$
This proves the lemma.
\end{proof}

\begin{remark} We note that Devroye and Lugosi (Chapter 7 of \cite{DL:01}) prove a similar result, but there are some differences. They also have all pairs of distributions in the cover compete against each other, but they use a different notion of competition between every pair. Moreover, their approach chooses a distribution in the cover that wins the
maximum number of competitions, whereas our algorithm chooses a distribution that is never defeated (i.e., won or tied against all other distributions in the cover).\end{remark}

\begin{remark}\costasnote{Recent work~\cite{DaskalakisK14,AcharyaJOS14,SureshOAJ14} improves the running time of the tournament approaches of Lemma~\ref{lem:log-cover-size}, Devroye-Lugosi and other related tournaments to have a quasilinear dependence of $O(N \log N)$ on the size $N=|{\cal S}_\eps|$ . In particular, they avoid running {\tt Choose-Hypothesis} for all pairs of distributions in ${\cal S}_\eps$.}
\end{remark}

\subsection{Proof of Theorem~\ref{thm:main}.} \label{sec:proofofmaintheorem}

\begin{figure}[h!]
\framebox{
\medskip \noindent \begin{minipage}{16cm}

\blue{

\medskip

{\tt Non-Proper-Learn-PBD}$(n,\eps,\delta)$

\begin{enumerate}

 \item Run {\tt Learn-Sparse}$^X(n,{\epsilon \over 12 \max\{c_1,c_2\}},\delta/3)$ to get hypothesis distribution $H_S$.

 \item Run {\tt Learn-Poisson}$^X(n,{\epsilon \over 12 \max\{c_1,c_2\}},\delta/3)$ to get hypothesis distribution $H_P.$

\item Run {\tt Choose-Hypothesis}$^X(H_S,{\widehat{H_P}},{\eps / 8},\delta/3)$.  If it
    returns $H_S$ then return $H_S$, and if it returns $\widehat{H_P}$ then return $H_P$.
\end{enumerate}

}

\end{minipage}}
\caption{{\tt Non-Proper-Learn-PBD}$(n,\eps,\delta)$.  The values $c_1,c_2$ are the absolute constants from Lemmas~\ref{lem:learnsparse} and~\ref{lem:learnbinomial}.  ${\widehat{H_P}}$ is defined in terms of $H_P$ as described in Definition \ref{def:hatHP}.}
\label{figure:non-proper-learn-pBD}
\end{figure}

We first show Part (1) of the theorem, where the learning algorithm may output any distribution over $[n]$ and not necessarily a PBD.  \blue{The algorithm for this part of the theorem, {\tt Non-Proper-Learn-PBD}, is given in Figure~\ref{figure:non-proper-learn-pBD}.}  This algorithm follows the high-level structure outlined in Figure~\ref{figure:learn-pBD} with the following modifications: (a) first, if the total variation distance to within which we want to learn $X$ is $\epsilon$, the second argument of both {\tt Learn-Sparse} and {\tt Learn-Poisson} is set to ${\epsilon \over 12 \max\{c_1,c_2\}}$, where $c_1$ and $c_2$ are respectively the constants from Lemmas~\ref{lem:learnsparse} and~\ref{lem:learnbinomial}; (b) the third step of {\tt Learn-PBD} is replaced by {\tt Choose-Hypothesis}$^X(H_S,{\widehat{H_P}},{\eps / 8},\delta/3)$, where ${\widehat{H_P}}$ is defined in terms of $H_P$ as described in Definition \ref{def:hatHP} below; and (c) if  {\tt Choose-Hypothesis}
 returns $H_S$, then {\tt Learn-PBD} also returns $H_S$, while if  {\tt Choose-Hypothesis} returns $\widehat{H_P}$, then {\tt Learn-PBD} returns $H_P$.

\begin{definition} \label{def:hatHP}
({\bf Definition of  ${\widehat{H_P}}$}:) ${\widehat{H_P}}$ is defined in terms of $H_P$ and the support of $H_S$ in three steps:

\begin{enumerate}

\item [(i)] for all points $i$ such that $H_S(i)=0$, we let ${\widehat{H_P}}(i)={{H_P}}(i)$;

\item [(ii)] for all points $i$ such that $H_S(i)\neq 0$, we describe in Appendix~\ref{app:poisson} an efficient deterministic algorithm that numerically approximates $H_P(i)$ to within an additive error of $\pm \eps/48s$, where $s=O(1/\eps^3)$ is the cardinality of the support of $H_S$. If $\widehat{H_{P,i}}$ is the approximation to $H_P(i)$ output by the algorithm, we set $\widehat{H_P}(i)={\max\{0,\widehat{H_{P,i}}-\eps/48s\}}$; notice then that $H_P(i)-\eps/24s \le \widehat{H_P}(i) \le H_P(i)$; finally,

\item [(iii)] for an arbitrary point $i$ such that $H_S(i)=0$, we set $\widehat{H_P}(i)=1-\sum_{j \neq i} \widehat{H_P}(j)$, to make sure that $\widehat{H_P}$ is a probability distribution.

\end{enumerate}

Observe that $\widehat{H_P}$ satisfies $\dtv({\widehat{H_P}},H_P)\le \eps/24$, and therefore $|\dtv({\widehat{H_P}},X) - \dtv(X,H_P)| \le \eps/24$. Hence, if $\dtv(X,H_P) \le {\eps \over 12}$, then $\dtv(X, \widehat{H_P}) \le {\eps \over 8}$ and, if $\dtv(X,\widehat{H_P}) \le {6\eps \over 8}$, then $\dtv(X, {H_P}) \le {\eps}$.
\end{definition}

We remark that the reason why we do not wish to use $H_P$ directly in {\tt Choose-Hypothesis} is purely computational. In particular, since $H_P$ is a translated Poisson distribution, we cannot compute its probabilities $H_P(i)$ exactly, and we need to approximate them. On the other hand, we need to make sure that using approximate values will not cause {\tt Choose-Hypothesis} to make a mistake. Our $\widehat{H_P}$ is carefully defined so as to make sure that {\tt Choose-Hypothesis} selects a probability distribution that is close to the unknown $X$, and that all probabilities that {\tt Choose-Hypothesis} needs to compute can be computed without much overhead. In particular, we remark that, in running {\tt Choose-Hypothesis}, we do not a priori compute the value of $\widehat{H_P}$ at every point; we do instead a lazy evaluation of $\widehat{H_P}$, as explained in the running-time analysis below.

We now proceed to the analysis of our modified algorithm {\tt Learn-PBD}. The sample complexity bound and correctness of our algorithm are immediate consequences of Lemmas~\ref{lem:learnsparse},~\ref{lem:learnbinomial} and~\ref{lem:choosehypothesis}, taking into account the precise choice of constants and the distance between $H_P$ and $\widehat{H_P}$. Next, let us bound the running time.
Lemmas~\ref{lem:learnsparse} and~\ref{lem:learnbinomial} bound the running time of Steps~1 and~2 of the algorithm, so it remains to bound the running time of the {\tt Choose-Hypothesis} step. Notice that ${\cal W}_1(H_S,\widehat{H_P})$ is a subset of the support of the distribution $H_S$. Hence to compute ${\cal W}_1(H_S,\widehat{H_P})$ it suffices to determine the probabilities $H_S(i)$ and $\widehat{H_P}(i)$ for  every point~$i$ in the support of $H_S$. For every such $i$, $H_S(i)$ is explicitly given in the output of {\tt Learn-Sparse}, so we only need to compute $\widehat{H_P}(i)$. It follows from Theorem~\ref{thm:poisson} (Appendix~\ref{app:poisson}) that the time needed to compute $\widehat{H_P}(i)$ is $\tilde{O}(\log(1/\eps)^3+\log(1/\eps) \cdot (\log n+\bit{\hat{\mu}} + \bit{\hat{\sigma}^2}))$. Since $\hat{\mu}$ and $\hat{\sigma}^2$ are output by {\tt Learn-Poisson}, by
 inspection of that algorithm it is easy to see that they each have bit complexity at most $O(\log n + \log(1/\eps))$ bits. Hence, given that the support of $H_S$ has cardinality $O(1/\eps^3)$, the overall time spent computing the probabilities $\widehat{H_P}(i)$ for  every point~$i$ in the support of $H_S$ is $\tilde{O}({1 \over \eps^3} \log n )$. After ${\cal W}_1$ is computed, the computation of the values $p_1=H_S({\cal W}_1)$, $q_1=\widehat{H_P}({\cal W}_1)$ and $p_1-q_1$ takes time linear in the data produced by the algorithm so far, as these computations merely involve adding and subtracting probabilities that have already been explicitly computed by the algorithm. Computing the fraction of samples from $X$ that fall inside ${\cal W}_1$ takes time  $O\left(\log n \cdot \log(1/\delta) / \eps^2 \right)$ and the rest of {\tt Choose-Hypothesis} takes time linear in the size of the
data that have been written down so far. Hence the overall running time of our algorithm is $\tilde{O}({1 \over \eps^3} \log n \log^2 {1 \over \delta})$.  This gives Part~(1) of Theorem~\ref{thm:main}.

\medskip
Now we turn to Part (2) of Theorem~\ref{thm:main}, the proper learning result.
\blue{The algorithm for this part of the theorem, {\tt Proper-Learn-PBD},
is given in Figure~\ref{figure:Proper-Learn-PBD}.  The algorithm is essentially
the same as {\tt Non-Proper-Learn-PBD} but with the following modifications, to produce a PBD that is within $O(\eps)$ of the unknown $X$:}  First, we replace {\tt Learn-Sparse} with a different learning algorithm,  {\tt Proper-Learn-Sparse}, which is based on Lemma~\ref{lem:log-cover-size}, and always outputs a PBD. Second, we add a post-processing step to {\tt Learn-Poisson} that converts the \costasnote{translated Poisson distribution} $H_P$ \costasnote{output by this procedure} to a PBD \blue{(in fact, to a Binomial distribution)}. After we describe these new ingredients in detail, we explain and analyze our proper learning algorithm.

\begin{figure}[h!]
\framebox{
\medskip \noindent \begin{minipage}{16cm}

\blue{

\medskip

{\tt Proper-Learn-PBD}$(n,\eps,\delta)$

\begin{enumerate}

 \item Run {\tt Proper-Learn-Sparse}$^X(n,{\epsilon \over 12 \max\{c_1,c_2\}},\delta/3)$ to get hypothesis distribution $H_S$.

 \item Run {\tt Learn-Poisson}$^X(n,{\epsilon \over 12 \max\{c_1,c_2\}},\delta/3)$ to get hypothesis distribution $H_P = TP(\hat{\mu},\hat{\sigma}^2)$.

\item Run {\tt Choose-Hypothesis}$^X(H_S,{\widehat{H_P}},{\eps / 8},\delta/3)$.

\begin{enumerate}

\item If it
    returns $H_S$ then return $H_S$.

\item Otherwise, if it returns $\widehat{H_P}$, then run {\tt Locate-Binomial}$(\hat{\mu}, \hat{\sigma}^2,n)$ to obtain a Binomial distribution $H_B = \mathrm{Bin}(\hat{n},\hat{p})$ with $\hat{n} \leq n$, and return
    $H_B$.

\end{enumerate}

\end{enumerate}

}

\end{minipage}}
\caption{{\tt Proper-Learn-PBD}$(n,\eps,\delta)$.  The values $c_1,c_2$ are the absolute constants from Lemmas~\ref{lem:learnsparse} and~\ref{lem:learnbinomial}.  ${\widehat{H_P}}$ is defined in terms of $H_P$ as described in Definition \ref{def:hatHP}.}
\label{figure:Proper-Learn-PBD}
\end{figure}

\begin{figure}[h!]
\framebox{
\medskip \noindent \begin{minipage}{16cm}

\blue{

\medskip

{\tt Proper-Learn-Sparse}$(n,\eps,\delta)$

\begin{enumerate}

\item Draw \costasnote{$M=32\log(8/\delta)/\eps^2$} samples from $X$ and sort them to obtain a list of values $0 \leq s_1 \leq \cdots \leq s_M \leq n.$

\item Define $\hat{a}:=s_{\lceil 2 \eps M \rceil}$ and $\hat{b}:=s_{\lfloor(1 - 2 \eps)M\rfloor}$.

\item If $\hat{b}-\hat{a} > (C/\eps)^3$ (where $C$ is the constant in the statement of Theorem~\ref{thm: sparse cover theorem}), output ``fail'' and return the (trivial)
  hypothesis which puts probability mass $1$ on the point $0$.

\item Otherwise,

\begin{enumerate}

\item Construct ${\cal S}'_\eps$, an $\eps$-cover of the set of all PBDs of
order $(C/\eps)^3$ (see Theorem~\ref{thm: sparse cover theorem}).

\item Let $\tilde{\cal S}_\eps$ be the set of all distributions of the form
$A(x-\beta)$ where $A$ is a distribution from ${\cal S}'_\eps$ and $\beta$
is an integer in the range $[\hat{a}-(C/\eps)^3,\dots,\hat{b}].$

\item Run the tournament described in the proof of Lemma~\ref{lem:log-cover-size} on $\tilde{\cal S}_\eps$, using confidence parameter $\delta/2$.  Return
    the (sparse PBD) hypothesis that this tournament outputs.

\end{enumerate}

\end{enumerate}

}

\end{minipage}}
\caption{{\tt Proper-Learn-Sparse}$(n,\eps,\delta)$.}
\label{figure:Proper-Learn-Sparse}
\end{figure}

\begin{enumerate}

\item  {\tt Proper-Learn-Sparse}$^X(n,\eps,\delta)$: This procedure draws $\tilde{O}(1/\eps^2) \cdot \log(1/\delta)$ samples from $X$, does $(1/\eps)^{O\left( \log^2(1/\eps) \right)} \cdot \tilde{O} \left( \log n  \cdot \log {1 \over \delta} \right)$ bit operations, and outputs a PBD $H_S$ in sparse form.  The guarantee is similar to that of {\tt Learn-Sparse}. Namely, if $X$ is $\eps$-close to some sparse form PBD $Y$ in the cover ${\cal S}_{\epsilon}$ of Theorem~\ref{thm: sparse cover theorem}, then, with probability at least $1-\delta$ over the samples drawn from $X$, $\dtv(X,H_S) \leq 6 \eps.$

The procedure {\tt Proper-Learn-Sparse}$^X(n,\eps,\delta)$ is given in
Figure~\ref{figure:Proper-Learn-Sparse}; we explain the procedure in tandem with a proof of correctness.
As in {\tt Learn-Sparse}, we start by truncating $\Theta(\eps)$ of the probability mass
from each end of $X$ to obtain a conditional distribution $X_{[\hat{a},\hat{b}]}$. In particular, we compute $\hat{a}$ and $\hat{b}$ as described in the beginning of the proof of Lemma~\ref{lem:learnsparse} (setting $\eps'=\eps$ and $\delta'=\delta$). Claim~\ref{claim:anchors are fine} implies that, with probability at least $1-\delta/2$, $X(\leq \hat{a}), 1-X(\leq \hat{b}) \in [3\eps/2,5\eps/2]$. (Let us denote this event by ${\cal G}$.) We distinguish the following cases:
\begin{itemize}
\item If $\hat{b}-\hat{a} > \omega = (C/\eps)^3$, where $C$ is the constant in the statement of Theorem~\ref{thm: sparse cover theorem},
the algorithm outputs ``fail,'' returning the trivial hypothesis that puts probability mass $1$ on the point $0$.
Observe that, if $\hat{b}-\hat{a} > \omega$ and $X(\leq \hat{a}), 1-X(\leq \hat{b}) \in [3\eps/2,5\eps/2]$, then $X$ cannot be $\eps$-close to a sparse-form distribution in the cover.
\item  If $\hat{b}-\hat{a} \le \omega$, {then the algorithm} proceeds as follows. Let ${\cal S}'_\eps$ be an $\eps$-cover of the set of all PBDs of order $\omega$, i.e., all PBDs which are sums of just
$\omega$ Bernoulli random variables. By Theorem~\ref{thm: sparse cover theorem}, it follows that
$|{\cal S}'_\eps|  = (1/\eps)^{O(\log^2(1/\eps))}$ and that ${\cal S}'_\eps$ can be constructed in time $(1/\eps)^{O(\log^2(1/\eps))}$.
Now, let $\tilde{{\cal S}}_\eps$ be the set of all distributions of the form $A(x - \beta)$ where $A$ is a
distribution from ${\cal S}'_\eps$ and $\beta$ is an integer ``shift'' which is in the range
$[\hat{a}-\omega, \ldots, \hat{b}]$.
Observe that there are $O(1/\eps^3)$ possibilities for $\beta$ and $|{\cal S}'_\eps|$ possibilities for $A$,
so we similarly get that $|\tilde{{\cal S}}_\eps|=(1/\eps)^{O(\log^2(1/\eps)}$ and that $\tilde{{\cal S}}_\eps$ can be constructed in time $(1/\eps)^{O(\log^2(1/\eps)} \log n$.
Our algorithm {\tt Proper-Learn-Sparse} constructs the set $\tilde{{\cal S}}_\eps$
and runs the tournament described in the proof of Lemma~\ref{lem:log-cover-size} (using $\tilde{{\cal S}}_\eps$ in place of ${\cal S}_\eps$, and $\delta/2$ in place of~$\delta$). We will show that, if $X$ is $\eps$-close to some sparse form PBD $Y \in {\cal S}_{\epsilon}$ and event ${\cal G}$ happens, then, with probability at least $1-{\delta \over 2}$, the output of the tournament is~a~sparse~PBD~that~is~$6\epsilon$-close~to~$X$.
\end{itemize}

{\bf Analysis:} The sample complexity and running time of {\tt Proper-Learn-Sparse} follow immediately from  Claim~\ref{claim:anchors are fine} and
Lemma~\ref{lem:log-cover-size}. To show correctness,
it suffices to argue that,  if $X$ is $\eps$-close to some sparse form PBD $Y \in {\cal S}_{\epsilon}$ and event ${\cal G}$ happens, then $X$ is $\eps$-close to some distribution in $\tilde{{\cal S}}_\eps$. Indeed, suppose that $Y$ is an order $\omega$ PBD $Z$ translated by some $\beta$ and suppose that $X(\leq \hat{a}), 1-X(\leq \hat{b}) \in [3\eps/2,5\eps/2]$. Since at least $1-O(\eps)$ of the mass of $X$ is in $[\hat{a}, \hat{b}]$, it is clear that $\beta$ must be in the range
$[\hat{a}-\omega, \ldots ,\hat{b}]$, as otherwise $X$ could not be $\eps$-close to $Y.$ So $Y \in \tilde{{\cal S}}_\eps$.

\begin{figure}[h!]
\framebox{
\medskip \noindent \begin{minipage}{16cm}

\blue{

\medskip

{\tt Locate-Binomial}$(\hat{\mu},\hat{\sigma}^2,n)$

\begin{enumerate}

\item If $ \hat{\sigma}^2 \le {n \over 4}$, set $\sigma_1^2=\hat{\sigma}^2$; otherwise, set $\sigma_1^2= {n \over 4}$.

\item If $\hat{\mu}^2 \le n (\hat{\mu}-\sigma_1^2)$, set $\sigma_2^2=\sigma_1^2$; otherwise,  set $\sigma_2^2 = {n \hat{\mu}- \hat{\mu}^2 \over n}$.

\item Return the hypothesis distribution $H_B = \mathrm{Bin}(\hat{n},\hat{p})$, where $\hat{n} = \left \lfloor \hat{\mu}^2 /(\hat{\mu} - \sigma^2_2) \right \rfloor$ and $\hat{p}={(\hat{\mu}-\sigma^2_2)/ \hat{\mu}}.$

    \end{enumerate}

}

\end{minipage}}
\caption{{\tt Locate-Binomial}$(\hat{\mu},\hat{\sigma}^2,n)$.}
\label{figure:Locate-Binomial}
\end{figure}

\item {\tt Locate-Binomial}$(\hat{\mu}, \hat{\sigma}^2,n)$: This routine takes as input the output $(\hat{\mu}, \hat{\sigma}^2)$ of {\tt Learn-Poisson}$^X(n,\eps,\delta)$ and computes a Binomial distribution $H_B$, without any additional samples from $X$. The guarantee is that, if $X$ is not $\epsilon$-close to any sparse form distribution in the cover $S_{\epsilon}$ of Theorem~\ref{thm: sparse cover theorem}, then, with probability at least $1-\delta$ (over the randomness in the output of {\tt Learn-Poisson}), $H_B$ will be $O(\epsilon)$-close to $X$.


Let $\mu$ and $\sigma^2$ be the (unknown) mean and variance of distribution $X$ and assume that  $X$ is not $\epsilon$-close to any sparse form distribution in $S_{\epsilon}$. Our analysis from Section~\ref{sec:kheavy} shows that, with probability at least $1-\delta$, the output $(\hat{\mu}, \hat{\sigma}^2)$ of {\tt Learn-Poisson}$^X(n,\eps,\delta)$ satisfies that $\dtv(X,TP(\hat{\mu},\hat{\sigma}^2))=O(\eps)$ as well as the bounds (\ref{eq:sigmabig}) and (\ref{eq:goodparams}) of Section~\ref{sec:kheavy} (with $\eps$ in place of $\eps'$). We will call all these conditions our ``working assumptions.'' We provide no guarantees when the working assumptions are not satisfied.

{\blue {\tt Locate-Binomial} is presented in Figure~\ref{figure:Locate-Binomial}; we proceed to explain the algorithm and establish its correctness.}  This routine has three steps. The first two eliminate corner-cases in the values of $\hat{\mu}$ and $\hat{\sigma}^2$, while the last step defines a Binomial distribution $H_B \equiv {\rm Bin}(\hat{n},\hat{p})$
    with $\hat{n} \leq n$ that is $O(\epsilon)$-close to $H_P\equiv TP(\hat{\mu},\hat{\sigma}^2)$ and hence to $X$ under our working assumptions. (We note that a significant portion of the work below is to ensure that $\hat{n} \leq n$, which does not seem to follow from a more direct approach.  Getting $\hat{n} \leq n$ is necessary in order for our learning algorithm for order-$n$ PBDs to  be truly proper.)  Throughout (a), (b) and (c) below we assume that our working assumptions hold. In particular, our assumptions are used every time we employ the bounds (\ref{eq:sigmabig}) and (\ref{eq:goodparams}) of Section~\ref{sec:kheavy}.

\begin{enumerate}
\item {\bf Tweaking $\hat{\sigma}^2$:} If $ \hat{\sigma}^2 \le {n \over 4}$, we set $\sigma_1^2=\hat{\sigma}^2$; otherwise, we set $\sigma_1^2= {n \over 4}$.  \blue{(As intuition for this tweak, observe that the largest possible variance of a Binomial distribution ${\rm Bin}(n,\cdot)$ is $n/4.$)}
    We note for future reference that in both cases~(\ref{eq:goodparams}) gives
\begin{equation}
\label{eq:sigma1}
(1-\eps)\sigma^2 \le \sigma_1^2 \leq (1 + \eps) \sigma^2,
\end{equation}
where the lower bound follows from (\ref{eq:goodparams}) and the fact that any PBD satisfies $\sigma^2 \le {n \over 4}$.

    We prove next that our setting of $\sigma_1^2$ results in
    $\dtv(TP(\hat{\mu},\hat{\sigma}^2),TP(\hat{\mu},\sigma_1^2)) \leq O(\eps).$  Indeed, if $\hat{\sigma}^2 \le {n \over 4}$ then this distance is zero and the claim certainly holds.  Otherwise we have that
$\left(1+{\epsilon} \right)\sigma^2 \ge \hat{\sigma}^2 > \sigma_1^2 = {n \over 4} \geq  \sigma^2,$ where we used~(\ref{eq:goodparams}).
Hence, by Lemma~\ref{lem: variation distance between translated Poisson distributions} we get:
\begin{eqnarray}\dtv(TP(\hat{\mu},\hat{\sigma}^2),TP(\hat{\mu},\sigma_1^2))
&\le& \frac{|\hat{\sigma}^2-\sigma_1^2|+1}{\hat{\sigma}^2} \nonumber \\
& \le & \frac{\epsilon\sigma^2+1}{\sigma^2} = O(\epsilon), \label{eq:a}
\end{eqnarray}
where we used the fact that $\sigma^2 = \Omega(1/\epsilon^2)$ from (\ref{eq:sigmabig}).

\item {\bf Tweaking $\sigma_1^2$:} If $\hat{\mu}^2 \le n (\hat{\mu}-\sigma_1^2)$ (\blue{equivalently,
$\sigma_1^2 \leq {\frac {n \hat{\mu} - \hat{\mu}^2} n}$}), set $\sigma_2^2=\sigma_1^2$; otherwise,  set $\sigma_2^2 = {n \hat{\mu}- \hat{\mu}^2 \over n}$.  \blue{(As intuition for this tweak, observe that the variance of a
${\rm Bin}(n,\cdot)$ distribution with mean $\hat{\mu}$ cannot exceed ${\frac {n \hat{\mu} - \hat{\mu}^2} n}.$)}
    We claim that this results in
       $\dtv(TP(\hat{\mu},{\sigma}_1^2),TP(\hat{\mu},\sigma_2^2)) \leq O(\eps).$ Indeed,
       if $\hat{\mu}^2 \le n (\hat{\mu}-\sigma_1^2)$, then clearly the distance is zero
       and the claim holds.  Otherwise
       \begin{itemize}
       \item Observe first that $\sigma_1^2 > \sigma_2^2$ and $\sigma_2^2 \ge 0$, where the last assertion follows from the fact that  $\hat{\mu} \le n$ by construction.

       \item Next, suppose that $X=PBD(p_1,\ldots,p_n)$. Then
from Cauchy-Schwarz we get that
\[
\mu^2 = \left(\sum_{i=1}^n p_i\right)^2 \leq n \left(\sum_{i=1}^n p_i^2 \right) = n(\mu - \sigma^2).
\]
Rearranging this yields
\begin{equation} \label{eq:CS}
{\frac {\mu(n-\mu)}{n}} \geq \sigma^2.
\end{equation}
We now have that
\begin{align}
\sigma_2^2 = {n \hat{\mu} - \hat{\mu}^2 \over n} &\ge {n (\mu-\eps \sigma) - ({\mu}+\eps \sigma)^2 \over n} \nonumber \\
&= {n \mu-\mu^2 -\eps^2 \sigma^2 - \eps \sigma(n+2 \mu ) \over n} \nonumber \\
&\ge \sigma^2 - {\eps^2 \over n} \sigma^2 - 3 \eps \sigma \notag\\
& \ge  (1-\epsilon^2) \sigma^2 - 3 \epsilon \sigma \ge (1-O(\eps))\sigma^2
\label{eq:blah}
\end{align}
where the first inequality follows from (\ref{eq:goodparams}), the second inequality follows
from (\ref{eq:CS}) and the fact that any PBD over $n$ variables satisfies $\mu \leq n,$ and the last one from~(\ref{eq:sigmabig}).
%
\item Given the above, we get by Lemma~\ref{lem: variation distance between translated Poisson distributions} that:
\begin{align}
\dtv(TP(\hat{\mu},\sigma_1^2),TP(\hat{\mu},\sigma_2^2)) \nonumber
&\le
\frac{\sigma_1^2-\sigma_2^2+1}{\sigma_1^2} \nonumber \\
 &\le \frac{(1+\epsilon) \sigma^2 - (1-O(\epsilon))\sigma^2+1}{(1-\epsilon)\sigma^2} = O(\epsilon), \label{eq:b} \end{align}
where we used that $\sigma^2 = \Omega(1/\epsilon^2)$ from~(\ref{eq:sigmabig}).
\end{itemize}

\item Constructing a Binomial Distribution: We construct a Binomial distribution $H_B$ that is $O(\epsilon)$-close to $TP(\hat{\mu},\sigma_2^2)$.
 If we do this then, by (\ref{eq:a}), (\ref{eq:b}), our working assumption that $\dtv(H_P,X)=O(\eps)$, and the triangle inequality,  we have that $\dtv(H_B, X) =O(\epsilon)$ and we are done. The
Binomial distribution $H_B$ that we construct is ${\rm Bin}(\hat{n},\hat{p})$,
where
$$
\hat{n} = \left \lfloor \hat{\mu}^2 /(\hat{\mu} - \sigma^2_2) \right \rfloor~~\text{and}~~\hat{p}={(\hat{\mu}-\sigma^2_2)/ \hat{\mu}}.$$
Note that, from the way that $\sigma_2^2$ is set in Step (b) above, we have that $\hat{n} \leq n$ and $\hat{p} \in [0,1]$, as required for  ${\rm Bin}(\hat{n},\hat{p})$ to be a valid Binomial distribution and a valid output for Part 2 of Theorem~\ref{thm:main}.

Let us bound the total variation distance between ${\rm Bin}(\hat{n},\hat{p})$ and $TP(\hat{\mu},\sigma_2^2)$. First, using Lemma~\ref{lem:translated Poisson approximation} we have:
\begin{eqnarray}
&&
\dtv({\rm Bin}(\hat{n},\hat{p}),TP(\hat{n}\hat{p}, \hat{n}\hat{p}(1-\hat{p}))
\nonumber \\
&\le& \frac{1}{\sqrt{\hat{n}\hat{p}(1-\hat{p})}}+\frac{2}{\hat{n}\hat{p}(1-\hat{p})}. \label{eq: lalalalaaa}
\end{eqnarray}
Notice that
\begin{eqnarray*}
\hat{n}\hat{p}(1-\hat{p}) &\ge&
\left ( {\hat{\mu}^2 \over \hat{\mu} - \sigma^2_2} -1\right )\left({\hat{\mu}-\sigma^2_2 \over \hat{\mu}}\right)\left({\sigma^2_2 \over \hat{\mu}}\right)\\
&=& \sigma_2^2 - \hat{p}(1-\hat{p}) \ge (1-O(\epsilon)) \sigma^2-1 \\
&\ge& \Omega(1/\eps^2),
\end{eqnarray*}
where the second inequality uses (\ref{eq:blah}) (or~\eqref{eq:sigma1} depending on which case of Step (b) we fell into) and the last one uses the fact that  $\sigma^2 = \Omega(1/\epsilon^2)$ from~\eqref{eq:sigmabig}. So plugging this into~\eqref{eq: lalalalaaa} we get:
$$
\dtv({\rm Bin}(\hat{n},\hat{p}) , TP(\hat{n}\hat{p}, \hat{n}\hat{p}(1-\hat{p})) =O(\epsilon).
$$

The next step is to compare $TP(\hat{n}\hat{p}, \hat{n}\hat{p}(1-\hat{p}))$
and $TP(\hat{\mu},\sigma_2^2)$. Lemma~\ref{lem: variation distance between
translated Poisson distributions} gives:
\begin{eqnarray*}
&&\dtv(TP(\hat{n}\hat{p}, \hat{n}\hat{p}(1-\hat{p})) , TP(\hat{\mu}, \sigma_2^2))\\
&\le& \frac{|\hat{n}\hat{p}-\hat{\mu}|}{\min(\sqrt{\hat{n}\hat{p}(1-\hat{p})},{\sigma_2})}+ \frac{|\hat{n}\hat{p}(1-\hat{p})-\sigma_2^2|+1}{\min(\hat{n}\hat{p}(1-\hat{p}), \sigma_2^2)} \\
&\le& \frac{1}{\sqrt{\hat{n}\hat{p}(1-\hat{p})}}+ \frac{2}{\hat{n}\hat{p}(1-\hat{p})}\\
& = &O(\eps).
\end{eqnarray*}
By the triangle inequality we get
$$\dtv({\rm Bin}(\hat{n},\hat{p}) , TP(\hat{\mu}, \sigma_2^2)=O(\epsilon),$$
which was our ultimate goal.
\end{enumerate}

\item {\tt Proper-Learn-PBD}:
Given the {\tt Proper-Learn-Sparse} and {\tt Locate-Binomial} routines described above, we are ready to describe our proper learning algorithm.
The algorithm is similar to our non-proper learning one, {\tt Learn-PBD}, with the following modifications:
In the first step, instead of running {\tt Learn-Sparse}, we run {\tt Proper-Learn-Sparse} to get a sparse form PBD $H_S$. In the second step, we still run {\tt Learn-Poisson} as we did before to get a translated Poisson distribution $H_P$. Then we run {\tt Choose-Hypothesis} feeding it $H_S$ and $H_P$ as input. If the distribution returned by {\tt Choose-Hypothesis} is $H_S$, we just output $H_S$. If it returns $H_P$ instead, then we run {\tt Locate-Binomial} to convert it to a Binomial distribution that is still close to the unknown distribution $X$. We tune the parameters $\epsilon$ and $\delta$ based on the above analyses to guarantee that, with probability at least $1-\delta$, the distribution output by our overall algorithm is $\epsilon$-close to the unknown distribution $X$. The number of samples we need is $\tilde{O}(1/\eps^2) \log(1/\delta)$, and the running time is $\left({1 \over \epsilon}\right)^{O(\log^2{1/\epsilon})} \cdot \tilde{O}(\log n \cdot \log {1 \over \delta})$. This concludes the proof of Part 2 of Theorem~\ref{thm:main}, and thus of the entire theorem. \qed

\end{enumerate}

\ignore{
}

\section{Learning weighted sums of independent Bernoullis}
\label{sec:constantnumofwts}

In this section we consider a generalization of the problem of learning an unknown PBD, by
studying the learnability of weighted sums of independent Bernoulli random variables
$X=\sum_{i=1}^n w_i X_i$.  (Throughout this section we assume for simplicity that the weights are ``known''
to the learning algorithm.)  In Section~\ref{sec:positive} we show that if there are only
constantly many different weights then such distributions can be learned by an algorithm that uses
$O(\log n)$ samples and runs in time $\poly(n).$  In Section~\ref{sec:negative} we show that if there are $n$ distinct weights then
even if those weights have an extremely simple structure -- the $i$-th weight is simply $i$ -- any algorithm must use $\Omega(n)$ samples.

\subsection{Learning sums of weighted independent Bernoulli random variables with few distinct weights}
\label{sec:positive}

Recall Theorem~\ref{thm:linearupper}:

\medskip

\noindent {\sc Theorem~\ref{thm:linearupper}.} \emph{Let $X = \sum_{i=1}^n a_i X_i$ be a weighted sum of unknown independent Bernoulli
random variables such that there are at most $k$ different values in the set
$\{a_1,\dots,a_n\}.$ Then there is an algorithm with the following properties:  given $n,$  $a_1,\dots,a_n$ and access to independent draws from $X$, it uses
$$\widetilde{O}(k/\eps^{2}) \cdot  \log(n) \cdot \log(1/\delta)$$
samples from the target distribution $X$,
runs in time
$$\poly \left( n^k \cdot (k/\eps)^{k\log^2(k/\eps)} \right) \cdot \log(1/\delta),$$
and with probability at least $1-\delta$ outputs a hypothesis vector $\hat{p} \in [0,1]^n$ defining independent Bernoulli random variables $\hat{X}_i$ with $\E[\hat{X}_i]=p_i$ such that $\dtv(\hat{X},X) \leq \eps,$ where $\hat{X}=\sum_{i=1}^n a_i \hat{X}_i$.}

\medskip

\blue{

\begin{remark}
\rocconote{A special case of a more general recent result} \cite{DDOST13} \costasnote{implies} a highly efficient algorithm for the special case of Theorem~\ref{thm:linearupper} in which the $k$ distinct values that $a_1,\dots,a_n$ can have are just $\{0,1,\dots,k-1\}$.  \costasnote{In this case, the} algorithm of \cite{DDOST13} draws $\poly(k,1/\eps)$ samples from the target distribution and, in the bit complexity model of this paper, has running time $\poly(k,1/\eps,\log n)$; thus its running time and sample complexity are both significantly better than Theorem~\ref{thm:linearupper}.  However, \rocconote{even the most general version of the \cite{DDOST13} result} \costasnote{cannot handle the full generality of Theorem~\ref{thm:linearupper}, which imposes} no conditions of any sort on the $k$ distinct weights --- they may be any real values. The \cite{DDOST13} result leverages known central limit theorems for total variation distance from probability theory that deal with sums of independent (small integer)-valued random variables.  We are not aware of such central limit theorems for the more general setting of arbitrary real values, and thus we take a different approach to Theorem~\ref{thm:linearupper}, via covers for PBDs, as described below.
\end{remark}
}

Given a vector $\overline{a}=(a_1,\dots,a_n)$ of weights,
we refer to a distribution $X=\sum_{i=1}^n a_i X_i$ (where $X_1,\dots,X_n$ are independent
Bernoullis which may have arbitrary means) as an \emph{$\overline{a}$-weighted sum of Bernoullis},
and we write ${\cal S}_{\overline{a}}$ to denote the space of all such distributions.

To prove Theorem~\ref{thm:linearupper} we first show that ${\cal S}_{\overline{a}}$ has an
$\eps$-cover that is not too large.  We then show that by running a ``tournament'' between
 all pairs of distributions in the cover, using the hypothesis testing subroutine from Section~\ref{sec:choosehypothesis}, it is possible to identify a distribution in the cover that is
 close to the target $\overline{a}$-weighted sum of Bernoullis.

\begin{lemma} \label{lem:smallcover}
There is an $\eps$-cover ${\cal S}_{\overline{a},\eps} \subset {\cal S}_{\overline{a}}$ of size
$|{\cal S}_{\overline{a},\eps}| \leq (n/k)^{3k} \cdot (k/\eps)^{k \cdot O(\log^2 (k/\eps))}$
that can be constructed in time $\poly(|{\cal S}_{\overline{a},\eps}|).$
\end{lemma}

\begin{proof}
Let $\{b_j\}_{j=1}^k$ denote the set of distinct weights in $a_1,\dots,a_n$, and let $n_j = \big|\{i \in [n] \mid a_i =b_j \}\big|$. With
this notation, we can write $X = \littlesum_{j=1}^k b_j S_{j} = g(S)$, where $S =(S_1, \ldots, S_k)$ with each $S_j$ a sum
of $n_j$ many independent Bernoulli random variables and $g(y_1, \ldots, y_k) = \sum_{j=1}^k b_jy_j$. Clearly we have
$\littlesum_{j=1}^k n_j = n$. By Theorem~\ref{thm: sparse cover theorem}, for each $j \in \{1,\dots,k\}$ the space of all possible $S_j$'s has an explicit $(\eps/k)$-cover
$\mathcal{S}^j_{\eps/k}$ of size $|\mathcal{S}^j_{\eps/k}| \le n_j^{{2}} + n \cdot (k/\eps)^{O(\log^2 (k/\eps))}$. By
independence across $S_j$'s, the product $\mathcal{Q} = \littleprod_{j=1}^k \mathcal{S}^j_{\eps/k}$ is an $\eps$-cover for the space of all possible $S$'s, and hence the set
\[
\{Q = \littlesum_{j=1}^k b_j S_j \ : \ (S_1,\dots,S_k) \in {\cal Q}\}
\]
is an $\eps$-cover for ${\cal S}_{\overline{a}}.$
So ${\cal S}_{\overline{a}}$ has an explicit $\eps$-cover of size $|\mathcal{Q}| = \littleprod_{j=1}^k
|\mathcal{S}^j_{\eps/k}| \leq (n/k)^{{2}k} \cdot (k/\eps)^{k \cdot O(\log^2 (k/\eps))}$.
\end{proof}

\begin{prevproof}{Theorem}{thm:linearupper}
We claim that the algorithm of Lemma~\ref{lem:log-cover-size} has the desired sample complexity and can be implemented to run in the claimed time bound.  The sample complexity bound follows directly from Lemma~\ref{lem:log-cover-size}.
It remains to argue about the time complexity. Note that the running time of the algorithm is $\poly(|{\cal S}_{\overline{a},\eps}|)$ times the running time of a competition. We will show that a competition between $H_1, H_2 \in {\cal S}_{\overline{a},\eps}$ can be carried out by an efficient algorithm.  This amounts to efficiently computing the probabilities $p_1 = H_1(\mathcal{W}_1)$ and $q_1 =
H_2(\mathcal{W}_1)$ {and efficiently computing $H_1(x)$ and $H_2(x)$ for each of the $m$ samples $x$ drawn in step (2) of the competition}. \blue{Note that each element $w \in \mathcal{W}$ (the support of
$X$ in the competition {\tt Choose-Hypothesis}) is a value $w = \sum_{j=1}^k b_j n'_j$ where $n'_j \in \{0,\dots,n_j\}.$}  \ignore{Note that $\mathcal{W} = \littlesum_{j=1}^k{b_i} \cdot \{0,1,\ldots,n_j \}$.} Clearly,
$|\mathcal{W}| \leq \littleprod_{j=1}^k (n_j+1)
= O((n/k)^k)$. It is thus easy to see that $p_1, q_1$ {and each of $H_1(x), H_2(x)$} can be efficiently
computed as long as there is an efficient algorithm for the following
problem: given $H=\sum_{j=1}^k b_jS_j \in {\cal S}_{\overline{a},\eps}$ and $w \in \mathcal{W}$,
compute $H(w)$.
Indeed, fix any such $H, w.$  We have
that
\[
H(w) = \sum_{m_1,\dots,m_k}
\littleprod_{j=1}^k \Pr_{H}[S_j = m_j],
\]
where the sum is over all $k$-tuples $(m_1,\dots,m_k)$ such that
$0 \leq m_j \leq n_j$ for all $j$ and $b_1 m_1 + \cdots + b_k m_k = w$
(as noted above there are at most $O((n/k)^k)$ such $k$-tuples).
To complete the proof of Theorem~\ref{thm:linearupper}
we note that $\Pr_{H}[S_j = m_j]$ can be computed in $O(n_j^2)$
time by standard dynamic programming. 
\end{prevproof}

We close this subsection with the following remark:  In \cite{DDS12:kmodallearn} the authors have given a $\poly(\ell,$ $\log(n),$ $1/\eps)$-time algorithm that learns any $\ell$-modal
distribution over $[n]$ (i.e., a distribution whose pdf has at most $\ell$ ``peaks'' and ``valleys'') using
$O(\ell \log(n)/\eps^3 + (\ell/\eps)^3 \log(\ell/\eps))$ samples.  It is natural to wonder whether this
algorithm could be used to efficiently learn a sum of $n$ weighted independent Bernoulli random variables
with $k$ distinct weights, and thus give an alternate algorithm for Theorem~\ref{thm:linearupper}, perhaps with better asymptotic guarantees.  However,   it is easy to
construct a sum $X=\sum_{i=1}^n a_i X_i$ of $n$ weighted independent Bernoulli random variables
with $k$ distinct weights such that $X$ is $2^k$-modal.  Thus, a naive application of the \cite{DDS12:kmodallearn}
result would only give an algorithm with sample complexity exponential in $k$, rather than
the quasilinear sample complexity of our current algorithm.  If the $2^k$-modality of the above-mentioned example is the worst case (which we do not know), then the \cite{DDS12:kmodallearn} algorithm would give
a $\poly(2^k,\log(n),1/\eps)$-time algorithm for our problem that uses $O(2^k \log(n)/\eps^3) +
2^{O(k)} \cdot \tilde{O}(1/\eps^3)$ examples (so comparing with Theorem~\ref{thm:linearupper}, exponentially worse sample complexity
as a function of $k$, but exponentially better running time as a function of $n$).  Finally, in the context of this question (how many modes can there be for a sum of $n$ weighted independent Bernoulli random variables with $k$ distinct weights), it is interesting to recall the result of K.-I. Sato \cite{Sato:93} which shows that for any $N$ there are two unimodal distributions $X,Y$ such that $X+Y$ has at least $N$ modes.


\subsection{Sample complexity lower bound for learning sums of weighted independent Bernoulli
random variables} \label{sec:negative}

Recall Theorem~\ref{thm:linearlower}:

\medskip

\noindent {\sc Theorem~\ref{thm:linearlower}.} \emph{Let
$X=\sum_{i=1}^n i \cdot X_i$ be a weighted sum of unknown independent Bernoulli random variables
(where the $i$-th weight is simply $i$). Let $L$ be any learning algorithm which, given $n$
and access to independent draws from $X$, outputs a hypothesis distribution $\hat{X}$ such that
$\dtv(\hat{X},X) \leq 1/25$ with probability at least
$e^{-o(n)}.$ Then $L$ must use $\Omega(n)$ samples.}

\medskip

\blue{The intuition underlying this lower bound is straightforward:  Suppose there are $n/100$
variables $X_i$, chosen uniformly at random, which have $p_i=100/n$ (call these the ``relevant variables''),
and the rest of the $p_i$'s are zero.
Given at most $c \cdot n$ draws from $X$ for a small constant $c$,
with high probability some constant fraction of the relevant $X_i$'s will not have been ``revealed''
as relevant, and from this it is not difficult to show that any hypothesis must have constant error.
A detailed argument follows.}

\begin{prevproof}{Theorem}{thm:linearlower}
We define a probability distribution over possible target probability distributions $X$ as follows:
A subset $S \subset \{n/2 + 1, \dots, n\}$ of size $|S|=n/100$ is
drawn uniformly at random from all ${n/2 \choose n/100}$ possible outcomes..
The vector $\overline{p}=(p_1,\dots,p_n)$ is defined as follows:
for each $i \in S$ the value $p_i$ equals $100/n = 1/|S|,$ and for
all other $i$ the value $p_i$ equals 0.  The $i$-th Bernoulli random variable $X_i$ has $\E[X_i]=p_i$, and
the target distribution is $X=X_{\overline{p}}=\sum_{i=1}^n i X_i.$

We will need two easy lemmas:

\begin{lemma} \label{lemma:l1}
Fix any $S,\overline{p}$ as described above.
For any $j \in \{n/2 + 1,\dots,n\}$ we have $X_{\overline{p}}(j) \neq 0$ if and
only if $j \in S$.  For any $j\in S$ the value  $X_{\overline{p}}(j)$ is exactly
$(100/n)(1 - 100/n)^{n/100 - 1} > 35/n$ (for $n$ sufficiently large),
and hence $X_{\overline{p}}(\{n/2+1,\dots,n\})>0.35$ (again for $n$ sufficiently large).

\end{lemma}
The first claim of the lemma holds because any set of $c \geq 2$ numbers from $\{n/2+1,\dots,n\}$
must sum to more than $n$.
The second claim holds because the only way a draw $x$ from $X_{\overline{p}}$ can have $x=j$ is if $X_j=1$ and all other $X_i$ are 0 (here we are using $\lim_{x \rightarrow \infty}(1 - 1/x)^{x} = 1/e$).

The next lemma is an easy consequence of Chernoff bounds:

\begin{lemma} \label{lemma:l2}
Fix any $\overline{p}$ as defined above, and consider a sequence of $n/2000$ independent draws from $X_{\overline{p}}
= \sum_{i} i X_i$.
With probability $1-e^{-\Omega(n)}$ the total number of indices
$j \in [n]$ such that $X_j$ is ever 1 in any of the $n/2000$ draws is at
most $n/1000$.
\end{lemma}

%
%

We are now ready to prove Theorem~\ref{thm:linearlower}.
Let $L$ be a learning algorithm that receives $n/2000$ samples.  Let $S \subset
\{n/2+1,\dots,n\}$ and $\overline{p}$ be chosen randomly
as defined above, and set the target to $X=X_{\overline{p}}.$

We consider an augmented learner $L'$ that is given ``extra information.''
For each point in the sample, instead of receiving the value of that draw from $X$
the learner $L'$ is given the entire vector
$(X_1,\dots,X_n) \in \{0,1\}^n$.  Let $T$ denote the set of elements $j \in \{n/2+1,\dots,n\}$ for which the learner is ever given a vector $(X_1,\dots,X_n)$ that has $X_j=1.$  By Lemma~\ref{lemma:l2} we have
$|T| \leq n/1000$ with probability at least $1 - e^{-\Omega(n)}$; we condition on the event
$|T| \leq n/1000$ going forth.

Fix any value $\ell \leq n/1000.$  Conditioned on $|T|=\ell,$
the set $T$ is equally likely to be any $\ell$-element subset of $S$,
and all possible ``completions'' of $T$ with
an additional $n/100-\ell \geq 9n/1000$ elements of $\{n/2+1,\dots,n\} \setminus T$
are equally likely to be the true set $S$.

Let $H$ denote the hypothesis distribution over $[n]$ that algorithm $L$ outputs.
Let $R$ denote the set $\{n/2 + 1,\dots,n\} \setminus T$; note that since
$|T|=\ell \leq n/1000$, we have $|R| \geq 499n/1000.$  Let $U$ denote the set
$\{i \in R \ : \ H(i) \geq 30/n\}.$  Since $H$ is a distribution we must have
$|U| \leq n/30.$  {It is easy to verify that we have $\dtv(X,H) \geq
{\frac 5 n} |S \setminus U|.$}
Since $S$ is a uniform random extension of $T$ with at most $n/100 - \ell \in [9n/1000,n/100]$ unknown elements of $R$ and $|R| \geq 499n/1000$, an easy calculation shows that
$\Pr[|S \setminus U| > 8n/1000]$ is $1 - e^{-\Omega(n)}$.  This means that with probability
$1 - e^{-\Omega(n)}$ we have $\dtv(X,H) \geq {\frac {8n}{1000}}\cdot {\frac {5}{n}} = 1/25$, and
the theorem is proved. \end{prevproof}

\section{Conclusion and open problems} \label{sec:conclusion}

{Since the initial conference
publication of this work \cite{DDS12pbdstoc},
some progress has been made on problems related to learning
Poisson Binomial Distributions.
The initial conference version \cite{DDS12pbdstoc} asked whether
log-concave distributions over $[n]$ (a generalization of
Poisson Binomial Distributions) can be learned to accuracy $\eps$ with $\poly(1/\eps)$
samples independent of $n$.  An affirmative answer to this question was subsequently provided
in \cite{CDSS13}.  More recently, \cite{DDOST13} studied a different
generalization of Poisson Binomial Distributions by considering
random variables of the form $X = \sum_{i=1}^n X_i$ where the $X_i$'s
are mutually independent (not necessarily identical) distributions
that are each supported on the integers $\{0,1,\dots,k-1\}$
(so, the $k=2$ case corresponds to Poisson Binomial Distributions).
\cite{DDOST13} gave an algorithm for learning these distributions
to accuracy $\eps$ using $\poly(k,1/\eps)$ samples (independent
of $n$).

While our results in this paper
essentially settle the sample complexity
of learning an unknown Poisson Binomial Distribution,
several goals remain for future work.  Our non-proper learning
algorithm is computationally more efficient than our proper learning
algorithm, but uses a factor of $1/\eps$ more samples.
An obvious goal is to obtain ``the best of both worlds'' by coming
up with an $O(1/\eps^2)$-sample algorithm which performs
$\tilde{O}(\log(n)/\eps^2)$ bit operations and learns an unknown PBD
to accuracy $\eps$ (ideally, such an algorithm
would even be proper and output a PBD as its
hypothesis).  Another goal is to sharpen the sample complexity
bounds of \cite{DDOST13}
and determine the correct polynomial dependence on $k$ and
$1/\eps$ for the generalized problem studied in that work.}

\ignore{
%
}

\bibliographystyle{alpha}

\bigskip

\bibliography{allrefs}

\appendix

\section{Extension of the Cover Theorem:
Proof of Theorem~\ref{thm: sparse cover theorem}} \label{appendix:proof of extended cover theorem}

\smallskip Theorem~\ref{thm: sparse cover theorem} is restating the main cover theorem (Theorem~1) of~\cite{DP:cover}, except that it claims an additional property, namely what follows the word ``finally'' in the statement of the theorem. (We will sometimes refer to this property as the {\em last part} of Theorem~\ref{thm: sparse cover theorem} in the following discussion.) Our goal is to show that the cover of~\cite{DP:cover} already satisfies this property without any modifications, thereby establishing Theorem~\ref{thm: sparse cover theorem}. To avoid reproducing the involved constructions of~\cite{DP:cover}, we will assume that the reader has some familiarity with them. Still, our proof here will be self-contained.

First, we note that the $\epsilon$-cover ${\cal S}_\epsilon$ of Theorem~1 of~\cite{DP:cover} is a subset of a larger ${\epsilon \over 2}$-cover ${\cal S}'_{\epsilon/2}$ of size $n^2+n\cdot (1/\epsilon)^{O(1/\epsilon^2)}$, which includes all the $k$-sparse and all the $k$-heavy Binomial PBDs (up to permutations of the underlying $p_i$'s), for some $k=O(1/\epsilon)$. Let us call ${\cal S}'_{\epsilon/2}$ the ``large ${\epsilon \over 2}$-cover'' to distinguish it from ${\cal S}_\epsilon$, which we will call the ``small $\epsilon$-cover.'' The reader is referred to Theorem~2 in~\cite{DP:cover} (and the discussion following that theorem) for a description of the large $\epsilon \over 2$-cover, and to Section 3.2 of~\cite{DP:cover} for how this cover is used to construct the small $\epsilon$-cover. In particular,  the small $\epsilon$-cover is a subset of the large ${\epsilon/2}$-cover, including only a subset of the sparse form distributions in the large ${\epsilon/2}$-cover. Moreover, for every sparse form distribution in the large $\epsilon/2$-cover, the small $\epsilon$-cover  includes at least one sparse form distribution that is $\epsilon/2$-close in total variation distance. 
Hence, if the large $\epsilon/2$-cover satisfies the last part of Theorem~\ref{thm: sparse cover theorem} (with $\epsilon/2$ instead of $\epsilon$ and ${\cal S}_{\epsilon/2}'$ instead of ${\cal S}_\epsilon$), it follows that the small $\epsilon$-cover also satisfies the last part of Theorem~\ref{thm: sparse cover theorem}.

\smallskip So we proceed to argue that, for all $\epsilon$, the large $\epsilon$-cover implied by Theorem~2 of~\cite{DP:cover} satisfies the last part of Theorem~\ref{thm: sparse cover theorem}. Let us first review how the large cover is constructed.  (See Section~4 of~\cite{DP:cover} for the details.) For every collection of indicators $\{X_i\}_{i=1}^n$ with expectations $\{\E[X_i]=p_i\}_i$, the collection is subjected to two filters, called the {\em Stage 1} and {\em Stage 2} filters, and described respectively in Sections 4.1 and 4.2  of~\cite{DP:cover}. Using the same notation as~\cite{DP:cover}, let us denote by $\{Z_i\}_i$ the collection output by the Stage 1 filter and by $\{Y_i\}_i$ the collection output by the Stage 2 filter. The collection $\{Y_i\}_i$ output by the Stage 2 filter satisfies $\dtv(\sum_i X_i,\sum_i Y_i)\le \epsilon$, and is included in the cover (possibly after permuting the $Y_i$'s). Moreover, it is in sparse or heavy Binomial form. This way, it is made sure that, for every $\{X_i\}_i$, there exists some $\{Y_i\}_i$ in the cover that is $\epsilon$-close and is in sparse or heavy Binomial form. We proceed to show that the cover thus defined satisfies the last part of Theorem~\ref{thm: sparse cover theorem}.

For $\{X_i\}_{i}$, $\{Y_i\}_{i}$ and $\{Z_i\}_{i}$ as above, let $(\mu, \sigma^2)$, $(\mu_Z, \sigma_Z^2)$ and $(\mu_Y, \sigma_Y^2)$ denote respectively the (mean, variance) pairs of the variables $X=\sum_i X_i$, $Z=\sum_i Z_i$ and $Y=\sum_i Y_i$. We argue first that the pair $(\mu_Z, \sigma_Z^2)$ satisfies $|\mu - \mu_Z| = O(\epsilon)$ and $|\sigma^2-\sigma_Z^2| = O(\epsilon \cdot (1+\sigma^2))$. Next we argue that, if the collection $\{Y_i\}_i$ output by the Stage 2 filter is in heavy Binomial form, then $(\mu_Y, \sigma_Y^2)$  satisfies $|\mu - \mu_Y| = {O(1)}$ and $|\sigma^2-\sigma_Y^2| = O(1 + \epsilon \cdot (1+\sigma^2))$, concluding the proof.
\begin{itemize}
\item Proof for $(\mu_Z, \sigma_Z^2)$: The Stage 1 filter only modifies the indicators $X_i$ with $p_i \in (0,1/k) \cup (1-1/k,1)$, for some well-chosen $k = O(1/\epsilon)$. For convenience let us define ${\cal L}_k=\{i~\vline~ p_i\in(0,1/k)\}$ and ${\cal H}_k=\{i~\vline~ p_i\in(1-1/k,1)\}$ as in~\cite{DP:cover}. The filter of Stage 1 rounds the expectations of the indicators indexed by ${\cal L}_k$ to some value in $\{0,1/k\}$ so that no single expectation is altered by more than an additive $1/k$, and the sum of these expectations is not modified by more than an additive $1/k$. Similarly, the expectations of the indicators indexed by ${\cal H}_k$ are rounded to some value in $\{1-1/k,1\}$. See the details of how the rounding is performed in Section 4.1 of~\cite{DP:cover}. Let us then denote by $\{p_i'\}_i$ the expectations of the indicators $\{Z_i\}_i$ resulting from the rounding. We argue that the mean and variance of $Z=\sum_i Z_i$ is close to the mean and variance of $X$. Indeed,
\begin{eqnarray}
|\mu - \mu_Z|&=&\left|\sum_i p_i - \sum_i p_i'\right| \nonumber \\
&=& \left|\sum_{i\in {\cal L}_k \cup {\cal H}_k} p_i - \sum_{i \in {\cal L}_k \cup {\cal H}_k} p_i'\right|\nonumber \\
& \le & O(1/k)=O(\epsilon). \label{eq: bound on the mean}
\end{eqnarray}
Similarly,
\begin{eqnarray*}
|\sigma^2 - \sigma_Z^2|&=&\left|\sum_i p_i(1-p_i) - \sum_i p_i'(1-p_i')\right|
\\
&\le& \left|\sum_{i\in {\cal L}_k} p_i(1-p_i) - \sum_{i \in {\cal L}_k} p_i' (1-p_i')\right| +  \left|\sum_{i\in {\cal H}_k} p_i(1-p_i) - \sum_{i \in {\cal H}_k} p_i' (1-p_i')\right|.
\end{eqnarray*}
We proceed to bound the two terms of the RHS separately. Since the argument is symmetric for ${\cal L}_k$ and ${\cal H}_k$ we only do ${\cal L}_k$.
We have
\begin{align*}
 \left|\sum_{i\in {\cal L}_k} p_i(1-p_i) - \sum_{i \in {\cal L}_k} p_i' (1-p_i')\right|
&= \left|\sum_{i\in {\cal L}_k} (p_i-p_i')(1-(p_i+p_i'))\right| \\
&= \left|\sum_{i\in {\cal L}_k} (p_i-p_i')-\sum_{i\in {\cal L}_k}(p_i-p_i')(p_i+p_i')\right| \\
&\le \left|\sum_{i\in {\cal L}_k} (p_i-p_i')\right|+\left|\sum_{i\in {\cal L}_k}(p_i-p_i')(p_i+p_i')\right| \\
&\le {1\over k}+\sum_{i\in {\cal L}_k} |p_i-p_i'|(p_i+p_i')\\
&\le {1\over k}+{1\over k} \sum_{i\in {\cal L}_k} (p_i+p_i')\\
&\le {1\over k}+{1\over k} \left(2 \sum_{i\in {\cal L}_k} p_i+1/k\right)\\
&{=} {1\over k}+{1\over k} \left({2 \over 1-1/k} \sum_{i\in {\cal L}_k} p_i (1-{1/k})+1/k\right)\\
&\le {1\over k}+{1\over k} \left({2 \over 1-1/k} \sum_{i\in {\cal L}_k} p_i (1-p_i)+1/k\right)\\
&\le {1\over k}+{1\over k^2}+{2\over k-1} \sum_{i\in {\cal L}_k} p_i (1-p_i).
\end{align*}
Using the above (and a symmetric argument for index set ${\cal H}_k$) we obtain:
\begin{align}
|\sigma^2 - \sigma_Z^2| \le{2\over k}+{2\over k^2}+{2\over k-1} \sigma^2 = O(\epsilon)(1+\sigma^2). \label{eq:bound on the variance}
\end{align}

\item Proof for $(\mu_Y, \sigma_Y^2)$: After the Stage 1 filter is applied to the collection $\{X_i\}_i$, the resulting collection of random variables $\{Z_i\}_i$ has expectations $p'_i \in \{0,1\} \cup [1/k,1-1/k]$, for all $i$. The Stage 2 filter has different form depending on the cardinality of the set ${\cal M}=\{i~|~p_i' \in [1/k,1-1/k]\}$. In particular, if $|{\cal M}| > k^3$ the output of the Stage 2 filter is in heavy Binomial form, while if $|{\cal M}| \le k^3$ the output of the Stage 2 filter is in sparse form. As we are only looking to provide guarantee for the distributions in heavy Binomial form, it suffices to  only consider  the former case next.
\begin{itemize}
\item $|{\cal M}| > k^3$: Let $\{Y_i\}_i$ be the collection produced by Stage 2 and let $Y=\sum_i Y_i$. Then Lemma 4 of~\cite{DP:cover} implies that
$$|\mu_Z- \mu_Y| = O(1)~~\text{and}~~|\sigma^2_Z - \sigma^2_Y| = O(1).$$
Combining this with~\eqref{eq: bound on the mean} and~\eqref{eq:bound on the variance} gives
$$|\mu - \mu_Y| = O(1)~~\text{and}~~|\sigma^2 - \sigma^2_Y| = O(1 + \epsilon \cdot (1+\sigma^2)).$$

\end{itemize}

\end{itemize}

This concludes the proof of Theorem~\ref{thm: sparse cover theorem}. \qed

\section{Birg\'{e}'s theorem:  Learning unimodal distributions} \label{ap:birge}

Here we briefly explain how Theorem~\ref{thm:Birge unimodal} follows from \cite{Birge:97}.  We assume that the reader is moderately familiar with the paper \cite{Birge:97}.

Birg\'{e} (see his Theorem~1 and Corollary~1) upper bounds the expected variation distance between the target distribution (which he denotes $f$) and the hypothesis distribution that is constructed by his algorithm (which he denotes $\hat{f}_n$; it should
be noted, though, that his ``$n$'' parameter denotes the number of samples used by the algorithm, while we will denote this by ``$m$'', reserving ``$n$'' for the domain $\{1,\dots,n\}$ of the distribution). More
precisely, \cite{Birge:97}  shows that this expected variation distance is at most that of the Grenander estimator (applied
to learn a unimodal distribution when the mode is known) plus a lower-order term.  For our Theorem~\ref{thm:Birge unimodal} we take Birg\'{e}'s ``$\eta$'' parameter to be $\eps$.  With this choice of $\eta,$ by the results of \cite{Birge:87,Birge:87b} bounding the expected error of the Grenander estimator, if $m=O(\log(n)/\eps^3)$ samples
are used in Birg\'{e}'s algorithm then the expected variation distance between the target distribution and his hypothesis distribution is at most $O(\eps).$
To go from expected error ${O(\eps)}$ to an ${O(\eps)}$-accurate
hypothesis with probability at least $1-\delta$, we run the above-described algorithm $O(\log(1/\delta))$ times so that with
probability at least $1-\delta$ some hypothesis obtained is ${O(\eps)}$-accurate.  Then we use our hypothesis testing procedure of Lemma~\ref{lem:choosehypothesis}, or, more precisely, the extension provided in Lemma~\ref{lem:log-cover-size}, to identify an $O(\eps)$-accurate hypothesis {from within this pool of $O(\log(1/\delta))$ hypotheses}.  (The use of Lemma~\ref{lem:log-cover-size} is why the running time of Theorem~\ref{thm:Birge unimodal} depends quadratically on $\log(1/\delta)$ {and why the sample complexity contains the second ${\frac 1 {\eps^2}} \log {\frac 1 \delta} \log \log {\frac 1 \delta}$ term.})

It remains only to argue that a single run of Birg\'{e}'s algorithm on a sample of size $m = O(\log(n)/\eps^3)$
can be carried out in $\tilde{O}(\log^2(n)/\eps^3)$ bit operations (recall that each sample is a
$\log(n)$-bit string).  His algorithm  begins by locating an $r \in [n]$ that approximately minimizes the value of his
function $d(r)$ (see Section~3 of \cite{Birge:97}) to within an additive $\eta = \eps$ (see Definition~3 of his
paper); intuitively this $r$ represents his algorithm's ``guess'' at the true mode of the distribution.  To locate
such an $r$, following Birg\'{e}'s suggestion in Section~3 of his paper, we begin by identifying two consecutive
points in the sample such that $r$ lies between those two sample points.  This can be done using $\log m$ stages
of binary search over the (sorted) points in the sample, where at each stage of the binary search we compute the two functions $d^-$ and $d^+$ and proceed in the appropriate direction.  To compute the function $d^-(j)$ at a given point $j$ (the computation of $d^+$ is analogous), we recall that $d^-(j)$ is defined as the maximum difference over $[1,j]$ between the empirical cdf and its convex minorant over $[1,j]$.  The convex minorant of the empirical cdf
(over $m$ points) can be computed in $\tilde{O}((\log n)m)$ bit-operations (where the $\log n$ comes from the fact that each sample point is an element of $[n]$), and then by enumerating over all points in the sample that lie in $[1,j]$
(in time $O((\log n)m)$) we can compute $d^-(j).$  Thus it is possible to identify two adjacent points in the sample
such that $r$ lies between them in time $\tilde{O}((\log n)m).$  Finally, as Birg\'{e} explains in the last paragraph of Section~3 of his paper, once two such points have been identified it is possible to again use binary search to find
a point $r$ in that interval where $d(r)$ is minimized to within an additive $\eta.$  Since the maximum difference
between $d^-$ and $d_+$ can never exceed 1, at most $\log(1/\eta)=\log(1/\eps)$ stages of binary search are required
here to find the desired $r$.

Finally, once the desired $r$ has been obtained, it is straightforward to output the final hypothesis (which Birg\'{e}
denotes $\hat{f}_n$).  As explained in Definition~3, this hypothesis is the derivative of $\tilde{F}^r_n$, which is essentially the convex minorant of the empirical cdf to the left of $r$ and the convex majorant of the empirical cdf
to the right of $r$.  As described above, given a value of $r$ these convex majorants and minorants can be computed in $\tilde{O}((\log n)m)$ time, and the derivative is simply a collection of uniform distributions as claimed.  This concludes our sketch of how Theorem~\ref{thm:Birge unimodal} follows from \cite{Birge:97}.

\ignore{
By performing a binary search over the $m$ points
of the sample, such an $r$ can be found in time $O(\log(m))$ times the time required to evaluate $d(\cdot)$
on a single input.}

\section{Efficient Evaluation of the Poisson Distribution} \label{app:poisson}

In this section we provide an efficient algorithm to compute an additive
approximation to  the Poisson probability mass function. It seems
that this should be a basic operation in numerical analysis, but we were not able to find
it explicitly in the literature. Our main result for this section is the following.

\begin{theorem} \label{thm:poisson}
There is an algorithm that, on input a rational number $\lambda >0$, and integers $k  \ge 0$ and $t>0$, produces an estimate $\widehat{p_k}$ such that
$$\left|\widehat{p_k} - p_k\right| \le {1 \over t},$$
where $p_k={\lambda^k e^{-\lambda} \over k!}$ is the probability that the Poisson distribution of parameter $\lambda$ assigns to integer $k$. The running time of the algorithm is $\tilde{O}(\bit{t}^3 + \bit{k}\cdot \bit{t} +\bit{\lambda} \cdot \bit{t})$.
\end{theorem}

\begin{proof}
Clearly we cannot just compute $e^{-\lambda}$, $\lambda^k$ and $k!$ separately, as this will take time exponential in the description complexity of $k$ and $\lambda$. We  follow instead an indirect approach. We start by rewriting the target probability as follows
$$p_k = e^{-\lambda + k \ln(\lambda)-\ln(k!)}.$$
Motivated by this formula, let $$E_k:=-\lambda + k \ln(\lambda)-\ln(k!).$$
Note that $E_k \leq 0$. Our goal is to approximate $E_k$ to within high
enough accuracy and then use this approximation to approximate $p_k$.

In particular, the main part of the argument involves an efficient algorithm to compute an approximation $\widehat{\widehat{E_k}}$ to $E_k$ satisfying
\begin{equation} \label{eqn:epsilon-hat-hat}
\Big|\widehat{\widehat{{E}_k}}-E_k \Big| \le {1 \over 4t} \le {1 \over 2t} - {1\over 8 t^2 }.
\end{equation}
This approximation will have bit complexity $\tilde{O}(\bit{k} +\bit{\lambda}+\bit{t})$ and
be computable in time $\tilde{O}(\bit{k} \cdot \bit{t}+\bit{\lambda}+\bit{t}^3)$.

\medskip We  show that if we had such an approximation, then we would be able to complete the proof.
For this, we claim that it suffices to approximate $e^{\widehat{\widehat{{E_k}}}}$ to within
an additive error  ${1\over 2t}$. Indeed, if $\widehat{p_k}$ were the result of this approximation, then we would have:
\begin{eqnarray*}
\widehat{p}_k &\le& e^{\widehat{\widehat{E_k}}}+{1 \over 2t} \\
&\le& e^{{E}_k + {1 \over 2t} - {1\over 8 t^2 }} + {1 \over 2t}\\ &\le& e^{{E}_k + \ln(1+ {1 \over 2t})} + {1 \over 2t} \\&\le& e^{E_k}\left(1+{1 \over 2t}\right) + {1 \over 2t} \le p_k +{1\over t};
\end{eqnarray*}
and similarly
\begin{eqnarray*}
\widehat{p}_k &\ge& e^{\widehat{\widehat{E_k}}}-{1 \over 2t} \\
&\ge& e^{{E}_k - ( {1 \over 2t} - {1\over 8 t^2 })} - {1 \over 2t}\\ &\ge&
e^{{E}_k - \ln(1+ {1 \over 2t})} - {1 \over 2t} \\ &\ge& e^{E_k}\Big/\left(1+{1 \over 2t}\right) - {1 \over 2t}\\ &\ge&  e^{E_k}\left(1-{1 \over  2t}\right) - {1 \over 2t} \ge p_k -{1\over t}.
\end{eqnarray*}
To approximate $e^{\widehat{\widehat{{E_k}}}}$ given ${\widehat{\widehat{{E_k}}}}$, we need the following lemma:

\begin{lemma} \label{lemma:exp-approx}
Let $\alpha \leq 0$ be a rational number. There is an algorithm that computes an estimate
$\widehat{e^{\alpha}}$ such that
$$ \left| \widehat{e^{\alpha}} - e^{\alpha} \right| \leq {1 \over 2t}$$
and has running time $\tilde{O}(\bit{\alpha}\cdot \bit{t}+\bit{t}^2).$
\end{lemma}

\begin{proof}
Since  $e^{\alpha} \in [0,1]$, the point of the additive grid $\{ {i \over4t} \}_{i=1}^{4t}$
closest to $e^{\alpha}$ achieves error at most $1/(4t)$.
Equivalently, in a logarithmic scale, consider the grid $\{ \ln{i \over4t} \}_{i=1}^{4t}$ and  let
$j^{\ast}:=\arg\min_j \left\{ \Big|\alpha - {\ln({j \over 4t})}\Big| \right\}$. Then, we have
that $$\left|{j^{\ast} \over  (4t)} - e^{\alpha}\right| \leq {1 \over 4t}.$$
The idea of the algorithm is to approximately identify the point $j^{\ast}$, by computing
approximations to the points of the logarithmic grid combined with a binary search procedure.
Indeed, consider the ``rounded'' grid $\{ \widehat{\ln{i \over4t}} \}_{i=1}^{4t}$ where each
$\widehat{\ln({i \over 4t})}$ is an approximation to  $\ln({i \over 4t})$ that is accurate
to within an additive ${1 \over 16t}$. Notice that, for $i=1,\ldots,4t$:
$$\ln\left({i +1 \over 4t}\right)-\ln\left({i \over 4t}\right) = \ln\left(1+{1 \over i}\right)
\ge \ln\left(1+{1 \over 4t}\right) > 1/8t.$$
Given that our approximations are accurate to within an additive $1/16t$, it follows that the rounded grid
$\{ \widehat{\ln{i \over4t}} \}_{i=1}^{4t}$ is  monotonic in $i$.

The algorithm does not construct the points of this grid explicitly, but adaptively as it needs them.
In particular, it performs a binary search in the set $\{1,\ldots, 4t\}$ to find the point
$i^{\ast} := \arg\min_i \left\{ \Big|\alpha - \widehat{\ln({i \over 4t})}\Big| \right\}$.
In every iteration of the search, when the algorithm examines the point $j$, it
needs to compute the approximation  $g_j = \widehat{\ln({j \over 4t})}$ and evaluate the distance
$|\alpha - g_j |$.  {It is known that the logarithm of a number $x$ with a binary fraction of $L$ bits and an exponent of $o(L)$ bits can be computed to within a relative error $O(2^{-L})$ in time $\tilde{O}(L)$~\cite{brentzeroes}.} It follows from this that $g_j$ has $O(\bit{t})$ bits and can be computed in time $\tilde{O}(\bit{t})$. The subtraction takes linear time, i.e., it uses $O(\bit{\alpha}+\bit{t})$ bit operations. Therefore, each step of the binary search can be done
in time $O(\bit{\alpha})+ \tilde{O}(\bit{t})$ and thus the overall algorithm has
$O(\bit{\alpha} \cdot \bit{t})+ \tilde{O}(\bit{t}^2)$ running time.

The algorithm outputs $i^{\ast} \over 4t$ as its final approximation to $e^{\alpha}$.
We argue next that the achieved error is at most an additive $1 \over 2t$.
Since the distance between two consecutive points of
the grid $\{ \ln{i \over4t} \}_{i=1}^{4t}$ is more than $1/(8t)$
and our approximations are accurate to within an additive $1/16t$,
a little thought reveals that $i^{\ast} \in \{j^{\ast}-1,j^{\ast},j^{\ast}+1\}$.
This implies that $i^{\ast} \over 4t$ is within an additive $1/2t$ of $e^{\alpha}$ as desired,
and the proof of the lemma is complete.
\end{proof}


Given Lemma~\ref{lemma:exp-approx}, we describe how we could approximate $e^{\widehat{\widehat{{E_k}}}}$ given ${\widehat{\widehat{{E_k}}}}$.
Recall that we want to output an estimate $\widehat{p_k}$ such that
$|\widehat{p_k} -e^{\widehat{\widehat{{E_k}}}}| \leq 1/(2t)$.
We distinguish the following cases:
\begin{itemize}
\item If $\widehat{\widehat{{E_k}}} \ge 0$, we output $\widehat{p_k}:=1$. Indeed, given that $\Big|\widehat{\widehat{{E}_k}}-E_k \Big| \le {1 \over 4t}$ and $E_k \le 0$, if  $\widehat{\widehat{{E_k}}} \ge 0$ then $\widehat{\widehat{{E_k}}} \in [0,{1 \over 4t}]$. Hence, because $t\ge 1$, $e^{\widehat{\widehat{{E_k}}}} \in [1,1+1/2t]$, so $1$ is within an additive $1/2t$ of the right answer.

\item Otherwise, $\widehat{p_k}$ is defined to be the estimate obtained by applying Lemma~\ref{lemma:exp-approx} for $\alpha:= \widehat{\widehat{E_k}}$.  Given the bit complexity of $\widehat{\widehat{E_k}}$, the running time of this procedure will be  $\tilde{O}(\bit{k} \cdot \bit{t} +\bit{\lambda}\cdot \bit{t} + \bit{t}^2)$.
\end{itemize}
Hence, the overall running time is $\tilde{O}(\bit{k} \cdot \bit{t} + \bit{\lambda}\cdot \bit{t}+\bit{t}^3)$.

\bigskip In view of the above, we only need to show how to compute $\widehat{\widehat{E_k}}$.
There are several steps to our approximation:
\begin{enumerate}
\item (Stirling's Asymptotic Approximation): Recall Stirling's asymptotic approximation (see e.g.,~\cite{Whittaker:80} p.193), which says that $\ln k!$ equals
$$k \ln(k) - k + (1/2)\cdot \ln(2\pi) +\sum_{j=2}^m { \frac{B_j \cdot (-1)^j}{j(j-1) \cdot k^{j-1}} }+O(1/k^m)$$
where $B_k$ are the Bernoulli numbers.
We define an approximation of $\ln{k!}$ as follows:
$$\widehat{\ln k!}: =k \ln(k) - k + (1/2)\cdot \ln(2\pi) +
\sum_{j=2}^{m_0} { \frac{B_j \cdot (-1)^j}{j(j-1) \cdot k^{j-1}} }$$
for $m_0:= O\left( \left\lceil {\bit{t} \over \bit{k}} \right\rceil +1\right).$

\item (Definition of an approximate exponent $\widehat{E_k}$):
Define $\widehat{E_k}:=-\lambda + k \ln(\lambda)-\widehat{\ln(k!)}$.
Given the above discussion, we can calculate the distance of $\widehat{E_k}$
to the true exponent $E_k$ as follows:
\begin{align}
|E_k - \widehat{E_k}| \le |\ln(k!)-\widehat{\ln(k!)}|  &\le O(1/k^{m_0})\\
&\le {1 \over 10t}. \label{eq: loss due to sterling and pi}
\end{align}
So we can focus our attention to approximating $\widehat{E_k}$.
Note that  $\widehat{E_k}$ is the sum of $m_0+2 = O({\log t \over \log k})$ terms. To approximate it within error
$1/(10t)$, it suffices to approximate each summand within an additive error of $O(1/(t \cdot \log t))$. Indeed, we so approximate each summand and our final approximation $\widehat{\widehat{E_k}}$ will be the sum
of these approximations. We proceed with the analysis:

\item (Estimating $2\pi$): Since $2\pi$ shows up in the above expression,
we should try to approximate it. It is known that the first $\ell$ digits of $\pi$
can be computed exactly in time $O(\log \ell \cdot M(\ell))$, where $M(\ell)$ is the time
to multiply two $\ell$-bit integers~\cite{Salamin,brent}. For example, if we use the
Sch\"onhage-Strassen algorithm for  multiplication~\cite{SS71}, we get
$M(\ell)=O(\ell \cdot \log \ell \cdot \log \log \ell)$. Hence, choosing
$\ell:=\lceil \log_2(12t \cdot \log t)\rceil$, we can obtain in time $\tilde{O}(\bit{t})$
an approximation $\widehat{2\pi}$ of $2\pi$ that has a binary fraction of $\ell$ bits and satisfies:
$$|\widehat{2\pi}-2\pi| \le 2^{-\ell}~~\Rightarrow~~
(1-2^{-\ell}) 2 \pi \le \widehat{2\pi} \le (1+2^{-\ell}) 2\pi.$$
Note that, with this approximation, we have
$$ \left| \ln(2\pi) - \ln(\widehat{2\pi}) \right| \leq \ln(1-2^{-\ell})\leq 2^{-\ell} \leq 1/(12t \cdot \log t).$$

\item (Floating-Point Representation): We will also need accurate approximations to $\ln{\widehat{2\pi}}$, $\ln k$ and $\ln \lambda$. We think of $\widehat{2\pi}$ and $k$ as multiple-precision floating point numbers base $2$. In particular,
\begin{itemize}
\item $\widehat{2\pi}$ can be described with a binary fraction of $\ell+3$ bits and a constant size exponent; and
\item $k \equiv 2^{\lceil \log k\rceil}\cdot {k \over 2^{\lceil \log k\rceil}}$ can be described with a binary fraction of $\lceil \log k \rceil$, i.e., $\bit{k}$, bits and an exponent of length $O( \log  \log k)$, i.e., $O(\log \bit{k})$.
\end{itemize}
Also, since $\lambda$ is a positive rational number, $\lambda={\lambda_1 \over \lambda_2}$, where $\lambda_1$ and $\lambda_2$ are positive integers of at most $\bit{\lambda}$ bits. Hence, for $i=1,2$, we can think of $\lambda_i$ as a multiple-precision floating point number base $2$ with a binary fraction of $\bit{\lambda}$ bits and an exponent of length $O(\log \bit{\lambda})$. Hence, if we choose $L = \lceil \log_2(12(3k+1)t^2 \cdot k \cdot \lambda_1 \cdot \lambda_2) \rceil = O(\bit{k}+\bit{\lambda}+\bit{t})$,  we can represent all numbers $\widehat{2\pi}, \lambda_1,\lambda_2, k$ as multiple precision floating point numbers with a binary fraction of $L$ bits and an exponent of $O(\log L)$ bits.

\item (Estimating the logs): It is known that the logarithm of a number $x$ with a binary fraction of $L$ bits and an exponent of $o(L)$ bits can be computed to within a relative error $O(2^{-L})$ in time $\tilde{O}(L)$~\cite{brentzeroes}. Hence, in time $\tilde{O}(L)$ we can obtain approximations $\widehat{\ln \widehat{2\pi}}, \widehat{\ln k}, \widehat{\ln{\lambda_1}}, \widehat{\ln{\lambda_2}}$ such that:
\begin{itemize}
\item $|\widehat{\ln k} - {\ln k}| \le 2^{-L} {\ln k} \le {1 \over 12(3k+1)t^2}$; and similarly
\item $|\widehat{\ln \lambda_i} - {\ln \lambda_i}| \le {1 \over 12(3k+1)t^2}$, for $i=1,2$;
\item $|\widehat{\ln \widehat{2\pi}} - {\ln \widehat{2\pi}}| \le {1 \over 12(3k+1)t^2}.$
\end{itemize}

\item (Estimating the terms of the series):
To complete the analysis, we also need to approximate each term of the form
$c_j = \frac{B_j}{j(j-1) \cdot k^{j-1}}$ up to an additive error of $O(1/(t \cdot \log t))$. We do this as follows: We compute the numbers $B_j$ and $k^{j-1}$ exactly, and we perform the division approximately.

Clearly, the positive integer  $k^{j-1}$ has description complexity $j \cdot \bit{k} = O(m_0 \cdot \bit{k}) = O(\bit{t}+\bit{k})$, since $j  = O(m_0)$.  We compute $k^{j-1}$ exactly using repeated squaring in time
$\tilde{O}(j \cdot \bit{k}) = \tilde{O}(\bit{t}+\bit{k})$.
It is known~\cite{Fillebrown:92} that the rational number $B_j$ has $\tilde{O}(j)$ bits and can be computed in $\tilde{O}(j^2) = \tilde{O}(\bit{t}^2)$ time. Hence, the approximate evaluation of the term $c_j$ (up to the desired additive error of $1/(t \log t)$)  can be done in $\tilde{O}(\bit{t}^2+\bit{k})$, by a rational division operation (see e.g.,~\cite{Knuth:81}).
The sum of all the approximate terms takes linear time, hence the approximate evaluation of the entire truncated series (comprising at most $m_0 \leq \bit{t}$ terms) can be done in $\tilde{O}(\bit{t}^3+\bit{k} \cdot \bit{t})$ time overall.

\medskip

Let $\widehat{\widehat{E_k}}$ be the approximation arising if we use all the aforementioned approximations. It follows from the above computations that
$$\Big |\widehat{\widehat{E_k}} - \widehat{E_k} \Big| \le {1 \over 10t}.$$

\item (Overall Error): Combining the above computations we get:
$$\Big |\widehat{\widehat{E_k}} - {E_k} \Big| \le {1 \over 4t}.$$

The overall time needed to obtain $\widehat{\widehat{E_k}}$ was
$\tilde{O}(\bit{k} \cdot \bit{t}+\bit{\lambda}+\bit{t}^3)$ and the proof of Theorem~\ref{thm:poisson} is complete. \qed
\end{enumerate}
\end{proof}

\end{document}